\documentclass[twocolumn,english,letterpaper,superscriptaddress]{revtex4-1}
\usepackage[T1]{fontenc}
\usepackage[latin9]{inputenc}
\usepackage{amsmath}
\usepackage{amsthm}
\usepackage{amssymb}
\usepackage{color}
\usepackage{epsfig}
\usepackage{graphicx}
\usepackage{xcolor}
\usepackage[normalem]{ulem}
\makeatletter
\@ifundefined{textcolor}{}
{%
 \definecolor{BLACK}{gray}{0}
 \definecolor{WHITE}{gray}{1}
 \definecolor{RED}{rgb}{1,0,0}
 \definecolor{GREEN}{rgb}{0,1,0}
 \definecolor{BLUE}{rgb}{0,0,1}
 \definecolor{CYAN}{cmyk}{1,0,0,0}
 \definecolor{MAGENTA}{cmyk}{0,1,0,0}
 \definecolor{YELLOW}{cmyk}{0,0,1,0}
}

\theoremstyle{plain}
\newtheorem{thm}{\protect\theoremname}
  \theoremstyle{definition}
  \newtheorem{defn}[thm]{\protect\definitionname}
  \theoremstyle{plain}
  
  \theoremstyle{plain}
  
  \theoremstyle{plain}
  \newtheorem{lem}[thm]{\protect\lemmaname}
\def\qed{\leavevmode\unskip\penalty9999 \hbox{}\nobreak\hfill
     \quad\hbox{\leavevmode  \hbox to.675em{%
               \hfil\vrule   \vbox to.675em%
               {\hrule width.6em\vfil\hrule}\vrule\hfil}}
     \par\vskip3pt}

\makeatother

\usepackage{babel}
  \providecommand{\corollaryname}{Corollary}
  \providecommand{\definitionname}{Definition}
  \providecommand{\factname}{Fact}
  \providecommand{\lemmaname}{Lemma}
\providecommand{\theoremname}{Theorem}

\newcommand{\cH}{\mathcal{H}}
\newcommand{\cN}{\mathcal{N}}
\newcommand{\cD}{\mathcal{D}}

\newcommand{\ket}[1]{ | #1 \rangle}

\newcommand{\be}{\begin{equation}}
\newcommand{\ee}{\end{equation}}
\newcommand{\bq}{\begin{eqnarray}}
\newcommand{\eq}{\end{eqnarray}}

\makeatother

\usepackage{babel}
\begin{document}

\title{Topological phase diagram of $\cD(S_3)$ induced by forbidding charges and fluxes}

\author{Anna K\'om\'ar}
\affiliation{Institute for Quantum Information and Matter and Walter Burke Institute for Theoretical Physics, California Institute of Technology, Pasadena, California 91125, USA}
\author{Olivier Landon-Cardinal}
\affiliation{Department of Physics, McGill University, Montr\'eal, Qu\'ebec, Canada H3A 2T8}

\begin{abstract}
We analyze phase transitions induced by forbidding charges and fluxes in $\cD(S_3)$, the simplest non-Abelian model among quantum doubles, a class of 2D spin lattice topological models introduced by Kitaev. Contrary to a topological quantum field theory, the lattice degrees of freedom allow to forbid charges and fluxes independently, resulting in a non-trivial effect on dyons. Forbidding charges and fluxes leads to only a subset of the original anyons remaining, forming a new theory. We interpret the processes the theory undergoes in terms of condensation, spontaneous symmetry breaking, splitting of particles. Mapping the complete phase diagram of $\cD(S_3)$, we find two distinct groups of phases: quantum doubles of subgroups of $S_3$, and a non-trivial emergent chiral phase, $SU(2)_4$.
\end{abstract}
\maketitle

\section{Introduction}

Understanding phase transitions and critical phenomena of physical systems is an important step in understanding properties of physical materials and fundamental physical processes. Classical phase transitions, e.g. magnetic transitions of the Ising model are well-understood, and are described by Landau's theory, using a free-energy description of the processes \cite{Sachdev11}. Therein the free energy, a local function of the state of the system, undergoes a change where it encounters a discontinuity (first-order phase transition), or a discontinuity of a higher-order $n$th derivative ($n$th-order phase transition).

Quantum phase transitions occur as a result of a change in the Hamiltonian of a quantum system, when slowly changing from one Hamiltonian to another the system passes through a critical point of non-analiticity. This could happen at a point where the ground state and the first excited state produce a level-crossing (the gap closes), or even when the gap becomes small enough that it would close \emph{in a limit}. The classic description of Landau's theory can be extended to provide a faithful description of these processes as well, through a quantum-classical mapping \cite{Sachdev11}.

An important concept of Landau's theory is the existence of a \emph{local order parameter}, that characterizes the state of the system. This is the main parameter of the free energy function, and the form of the function is inferred using existing symmetries of the system. However, introduce topological theories, and Landau's theory breaks down. Topological order, an inherently non-local property, can't be captured through a local order parameter. As a result, changes in the topological order of a field theory can't be described using this classic framework. An accurate description of topological phase transitions would help us understand more about the nature of topological order.

One well-known theory exhibiting topological order is the quantum double construction, introduced by Kitaev in Ref.~\cite{Kitaev03}. This is a family of 2D spin theories on a lattice, with a Hamiltonian whose ground state exhibits topological order. This means that there exists no local parameter that characterizes the ground state, and as such, quantum doubles are prime candidates for storing quantum information in: information stored in their degenerate ground space won't decohere due to local noise from the environment. The specific quantum double, and thus the specific topological order such a theory realizes depends on the group, $G$ that is used to build the quantum double. As a result, quantum doubles built around different groups would have different topological properties.

In this work, we analyze phase transitions of quantum doubles, induced by changing the set of allowed excitations (anyons) in the theory. A correspondence between the set of anyons in a field theory, and the topological order of the ground state has been established in Ref.~\cite{LWW15}. Therefore, anyons of a theory may be treated as a signature of topological order, a concept we will revisit in the Discussion section of this paper.

We focus our analysis on the simplest non-Abelian quantum double, $\cD(S_3)$, and investigate what processes this double undergoes when we remove (or forbid) certain anyons from the theory, i.e. we don't allow their creation either through fusion of other particles, or through thermal processes from the vacuum. More precisely, we forbid conjugacy classes and irreducible representations of the group, which leads to forbidding anyons (or parts of their internal Hilbert spaces). We interpret the mathematical steps of our calculations as physical processes: condensation, (spontaneous) symmetry breaking, and splitting of particles.

We make an exhaustive map of the phase diagram of $\cD(S_3)$ (for an illustration, see Fig.~\ref{fig:phase_diagram_small}). Among the new phases we find quantum doubles of subgroups of $S_3$, as well as the theory $SU(2)_4$. While the emergence of theories based on subgroups is to be expected (as we decrease the set of allowed anyons), the fact that the double $\cD(S_3)$ can transition into $SU(2)_4$ is intriguing. It is especially interesting in light of a recent proposal \cite{LBF+15} to utilize the anyons of $SU(2)_4$ to build a universal gate-set for quantum computation.

\begin{figure}
\begin{centering}
\includegraphics[width=0.4\textwidth]{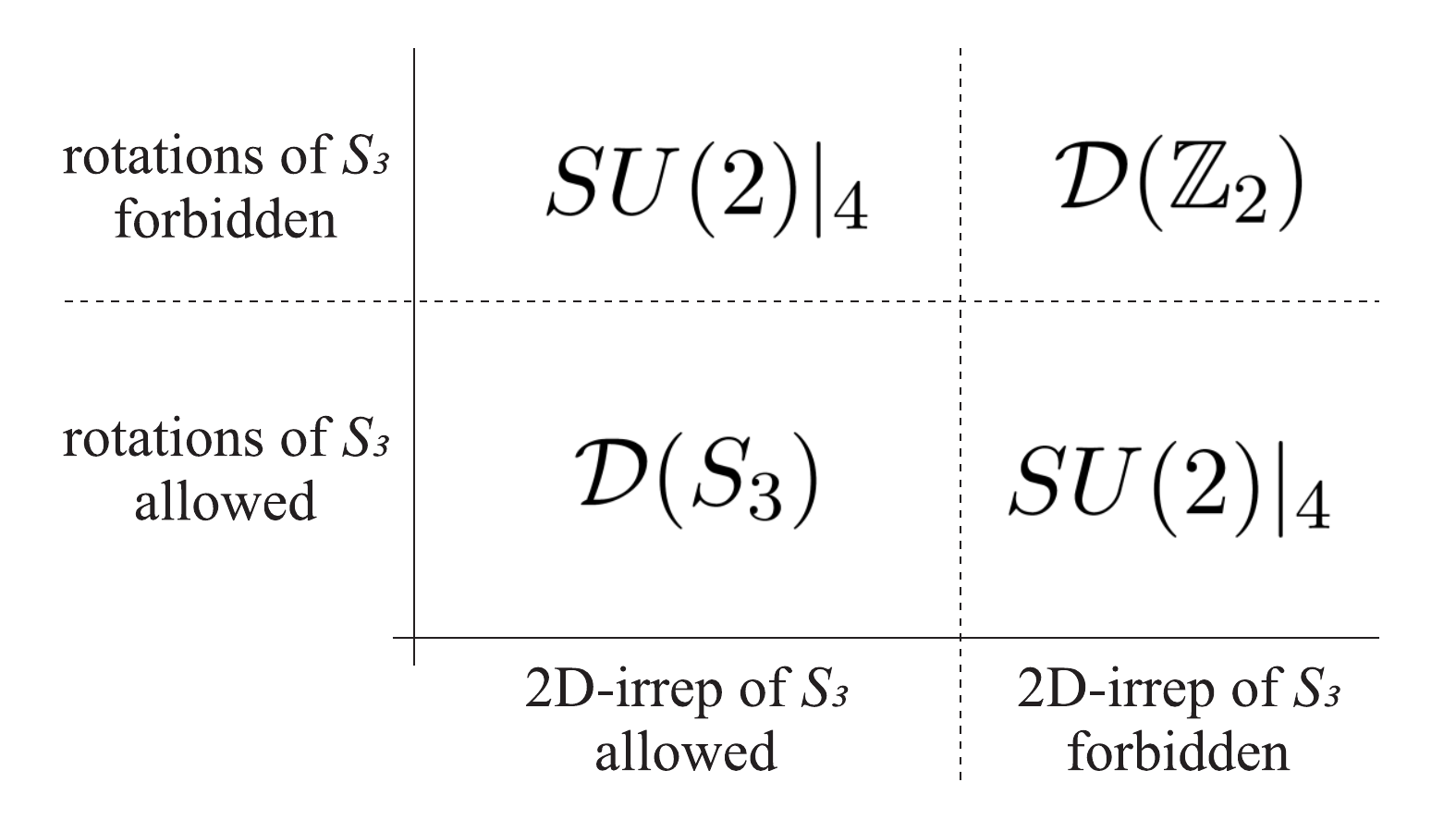}
\par\end{centering}
\caption[Part of the phase diagram of $\cD(S_3)$]{Part of the phase diagram of $\cD(S_3)$, showing the emergent theories when forbidding the 2d-irrep of $S_3$ (charge label) and/or the rotations of $S_3$ (flux label).}
\label{fig:phase_diagram_small}
\end{figure}

This paper is organized as follows. In Sec.~\ref{sec:framework} we provide a general review of anyons and anyonic data (Secs.~\ref{sec:anyons}-\ref{subsec:anyonic_data}) and introduce the Abelian and non-Abelian theories that appear in this paper (Sec.~\ref{sec:Drinfeld_double}). Then, in Sec.~\ref{sec:phase_transitions_general} we will give an overview of the physical processes a quantum double might undergo during a phase transition; this is in fact a summary of all processes found in the later sections, with their physical interpretations. In this section, we also give a detailed description of our mathematical protocol, used to obtain the results of this paper. In Sec.~\ref{sec:phase_diagram}, we demonstrate this protocol on three cases of phase transitions: one transition yields $\cD(\mathbb{Z}_3)$, while the other two result in the chiral model $SU(2)_4$. Until this point our analysis focuses only on the anyons, and the field theory aspects of quantum doubles. In Sec.~\ref{sec:lattice_picture} we investigate possible ways to realize these phase transitions, using the standard lattice description. We argue that while strictly speaking a global projection might be necessary to induce a phase transition, this can be approximated by adding local terms to the topological Hamiltonian of the model. Sec.~\ref{sec:conclusions} is a summary of our results and further directions. Finally, the appendices contain additional mathematical properties of the doubles discussed in this paper (Appendix~\ref{sec:fusion_Smatrix_DS3}), and the complete analysis of all phases of $\cD(S_3)$ that are not presented in the main text (Appendix~\ref{sec:all_theories}).

\section{Anyons in field theories} \label{sec:framework}

We start by introducing the concept of anyons, and that of non-Abelian anyons (Sec.~\ref{sec:anyons}). Anyons can emerge as excitations of two-dimensional field theories, their exchange statistics differ from the trivial statistics of bosons and fermions. In Sec.~\ref{subsec:anyonic_data} we will present the mathematical framework for the description of anyons, including the concept of the braiding $S$-matrices and fusion rules. Then, in Sec.~\ref{sec:Drinfeld_double} we introduce a few theories with anyonic excitations. We provide four examples: Abelian theories $\mathcal{D}(\mathbb{Z}_2)$, widely known as the toric code, and $\mathcal{D}(\mathbb{Z}_3)$, as well as the non-Abelian theories $SU(2)_4$ and $\mathcal{D}(S_3)$, the latter being the central object of this paper. 

\subsection{Anyons and labeling}
\label{sec:anyons}

In two-dimensional field theories, excitations can have esoteric exchange statistics, unlike those of bosons and fermions. Denoting the (counterclockwise) exchange operation of particles $a$ and $b$ by $R_{ab}$, doubly exchanging said particles of a two-dimensional theory in general can yield
\be
R_{ab}^2 \ket{a,b} = U_{a,b} \ket{a,b} ,
\ee
where $U_{a,b}$ is a unitary transformation acting on the joint wave function. If the effect of $U_{a,b}$ is a simple (non-trivial) phase: $\exp (i \varphi)$, then we call $a$ and $b$ \emph{Abelian anyons}, otherwise they are \emph{non-Abelian}.

Following Ref.~\cite{Preskill98}, we can introduce flux- and charge-labels for the anyons. The flux labels are elements of a group $g \in G$, and the charge labels are irreducible representations, or \emph{irreps} of the same group $G$. These two labellings, in fact, provide two complementary resolutions of the Hilbert space of an anyon:
\bq
\mathcal{H} &=& \textrm{span}\left\{ \ket{g} \; | \; g \in G \right\} \\
\mathcal{H} &=& \textrm{span}\left\{ \ket{\Gamma,i} \; | \; \Gamma \textrm{ irrep of } G, \; i \in \{1, ... |\Gamma|\} \right\} .
\eq

Anyons that only have a non-trivial flux, or a non-trivial charge label, are called \emph{fluxons} or \emph{chargeons}, while their other label is trivial. Anyons with both a non-trivial flux and charge label are called \emph{dyons}.

Due to the transformations fluxons undergo when braided with each other, it is not possible to use a group element $g$ as a globally agreed upon label for a fluxon \cite{Preskill98}. Rather, the gauge-invariant flux labels are conjugacy classes $C_g$ of $G$:
\begin{defn}[Conjugacy class]
\be
C_g = \{ zgz^{-1} | g,z \in G \} .
\ee
\end{defn}

We can find this flux label of an anyon by braiding the unknown anyon with known chargeons (labeled by irreps of $G$), i.e. transporting known chargeons around the unknown fluxon (see Fig.~\ref{fig:double-exchange}). We can similarly find the charge label of an anyon by braiding it with known fluxons. However, when we try to find the charge label of a dyon this way, we don't have access to the specific label $\Gamma^G$ (an irrep of $G$) \cite{Preskill98}. Rather, the gauge-invariant charge label of a dyon that has flux $g$: is an irrep of the normalizer group of $g$.

\begin{figure}
\begin{centering}
\includegraphics[width=0.3\textwidth]{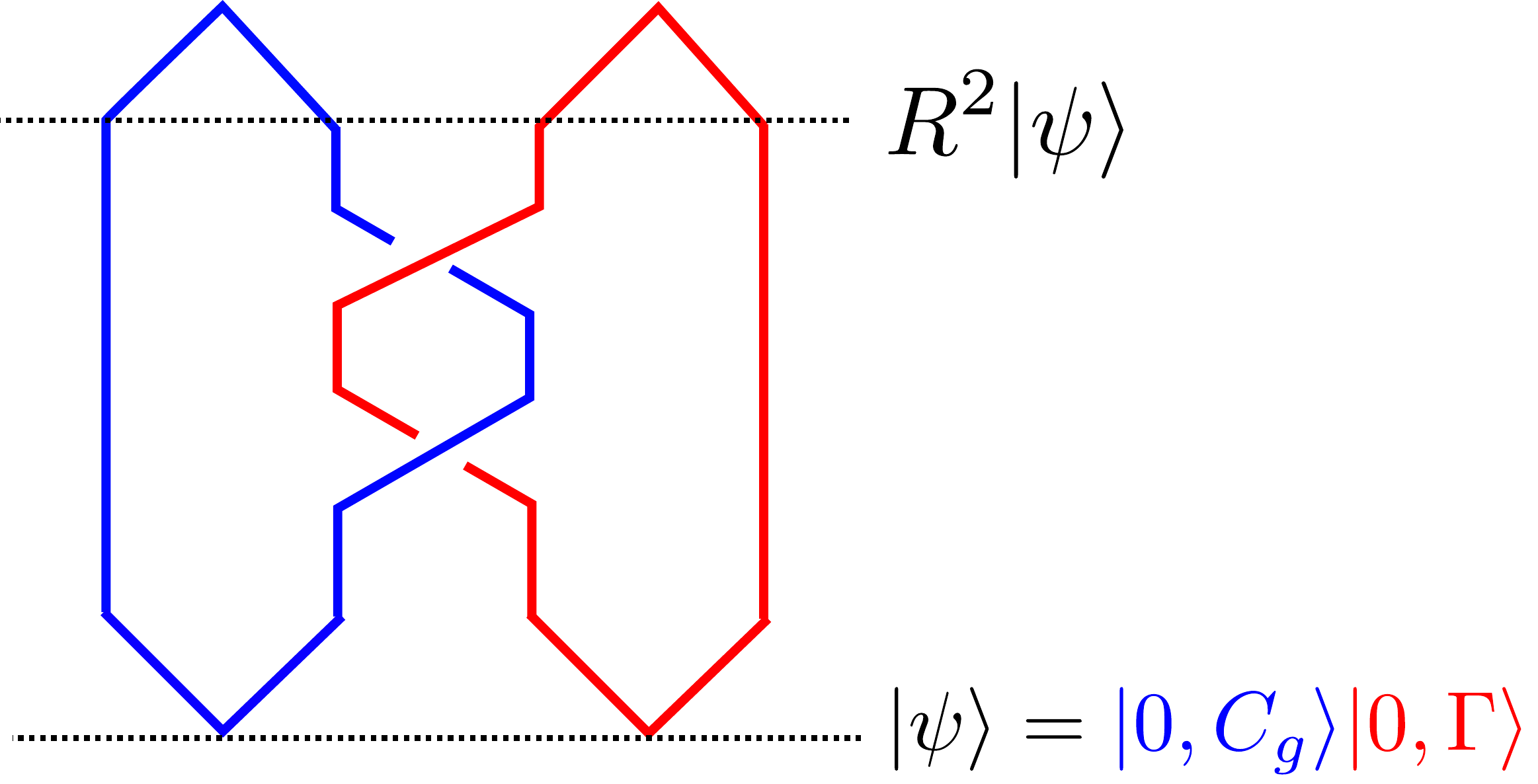}
\par\end{centering}
\caption{Braiding fluxon $C_g$ with chargeon $\Gamma$.}
\label{fig:double-exchange}
\end{figure}

\begin{defn}[Normalizer]
\be
\cN_h = \{ z \in G | zh = hz \} .
\ee
\end{defn}

The normalizer of an element of $G$ is always a subgroup of $G$. Furthermore, even though the exact normalizer group will depend on the specific choice of element $z$, the normalizer groups $\cN_z$ are isomorphic to each other for all $z \in C_g$.

\subsection{Anyonic data}
\label{subsec:anyonic_data}

Having reviewed how to label anyons, we would like to introduce the mathematical objects which encode the information about the anyon model, in particular the $S$-matrix, the fusion rules and the quantum dimensions.

\subsubsection{S-matrix}

The $S$-matrix encapsulates the braiding relations of the different anyons. Formally, the $S$-matrix elements are the amplitude probabilities of braiding events described in Fig.~\ref{fig:S-matrix}.

\begin{figure}
\begin{centering}
\includegraphics[width=0.4\columnwidth]{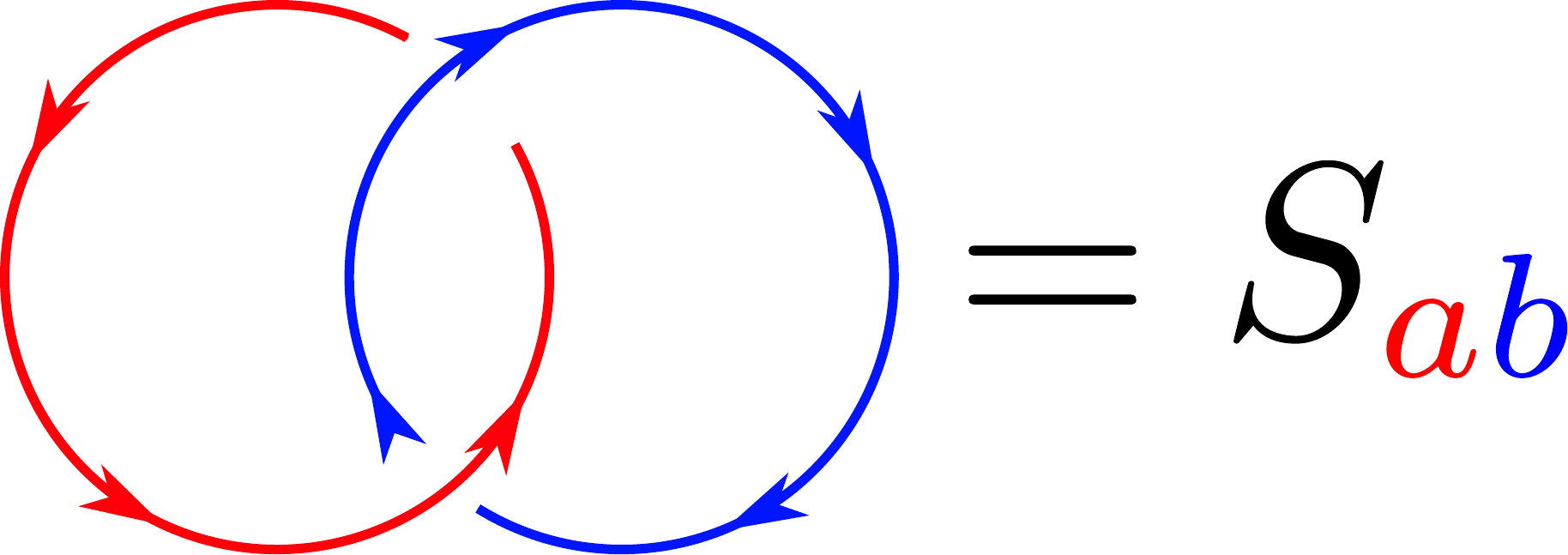}
\par\end{centering}
\caption[Definition of an $S$-matrix element]{Definition of the $S$-matrix element $S_{ab}$ where $a$ and $b$ label different anyon types. $S_{ab}$ is the probability amplitude associated to the braiding described by the figure.}
\label{fig:S-matrix}
\end{figure}

The $S$-matrix is unitary and for the models in the focus of this paper (quantum doubles) it can be computed from the representation theory of the underlying group thanks to the formula~\cite{BSW11,Coquereaux12}
\begin{eqnarray} \label{eq:Smatrix_doubles}
S_{(C_g,\Gamma),(C_{g'},\Gamma^{'})}  = &&  \tfrac{1}{\left|\mathcal{N}(g)\right|\left|\mathcal{N}(g')\right|} \\
&& \sum_{h:hg'h^{-1}\in\mathcal{N}(g)} \chi_{\Gamma}(hg'h^{-1}) \chi_{\Gamma'}(h^{-1}g^{-1}h), \nonumber
\end{eqnarray}
where $\chi_{\Gamma}(z) = \mbox{Tr} (\Gamma (z))$ is the character of representation $\Gamma$.

For completeness, let us state that the $S$-matrix is not sufficient to uniquely identify an anyon model. For instance, the quantum double of the dihedral group of degree 4, $D_4$ and the quantum double of the quaternion group $Q$ have the same $S$-matrices yet different anyon models. The $S$-matrix of $\mathcal{D}(D_4)$ is worked out in~\cite{FML+16} while the $S$-matrix of $\mathcal{D}(Q)$ is worked out in~\cite{XL12}. One also needs to provide the $T$-matrix which encapsulates the topological spin of each anyon. For quantum doubles, i.e. most theories appearing in this paper, the $T$-matrix can be computed thanks to 
\begin{equation}
T_{(C_g,\Gamma),(C_{g'},\Gamma^{'})}=\delta_{C_{g}C_{g'}}\delta_{\Gamma\Gamma'}\frac{\chi_{\Gamma}(g)}{\chi_{\Gamma}(e)} ,
\end{equation}
where $e$ is the identity element of group $G$.

\subsubsection{Fusion rules}

In addition to braiding, fusing two anyons is an essential feature of an anyon model. Suppose two anyons are next to one another. One would like to treat them as a single anyon, and identify the properties of this resulting anyon. The resulting anyon is the \emph{fusion} of the two original anyons. In the absence of information about the particular state of the two original anyons, the best description of the fusion state is summarized by listing the allowed fusion channels. This information is captured in the $N$-symbols:
\be \label{eq:N_symbol_definition}
a \times b = \sum_c N_{ab}^c c ,
\ee
where $N_{ab}^c$ is the multiplicity of particle $c$ when fusing anyon $a$ with anyon $b$. In a non-Abelian theory, at least one fusion has two (or more) possible fusion channels, making the result of that fusion non-deterministic. 

An important feature of the fusion rules is that for any anyon type $a$ there exists a unique anyon type $a^{-1}$ such that the vacuum is a possible fusion outcome. The corresponding anyon is called the antiparticle of $a$. 

Having defined an antiparticle, one can capture all the information about the braiding of anyons into the topological $S$-matrix whose elements are the amplitudes of the different braiding processes:
\be
S_{ab} = \text{Tr} (R_{b^{-1}a} R_{a b^{-1}}) / \mathcal{D} ,
\ee
where $b^{-1}$ is the antiparticle of anyon $b$. Eq.~\eqref{eq:Smatrix_doubles} in the previous section was a special case of this formula.

Fusing and braiding are not independent but related by the Verlinde-formula:
\be \label{eq:verlinde}
N_{ab}^c = \sum_z \frac{S_{az} S_{bz} S_{c^{-1}z}}{S_{1z}} .
\ee
where $1$ is the trivial anyon (vacuum). This relates the $S$-matrix to the $N$-symbols.

\subsubsection{Quantum dimensions}

The quantum dimension $d_a$ of anyon $a$ is the dimension of the space of anyon $a$, and is related to the number of internal degrees of freedom the anyon has. The fact that an anyon might have more than one internal degree of freedom, we will refer to as that anyon having different \emph{flavors}. Abelian anyons necessarily have $d_a = 1$, while non-Abelian anyons have quantum dimension $d_a > 1$.

The definition of quantum dimensions is through the fusion rules: the fusion space of a pair of anyons needs to be equal to the combined space of the outcome-anyons, taking their multiplicity into account. I.e. a fusion of anyons $a$ and $b$ in the form of Eq.~\eqref{eq:N_symbol_definition} means the quantum dimensions of participating anyons are \cite{Pachos12}:
\be
d_a \cdot d_b = \sum_c N_{ab}^c d_c .
\ee
Quantum dimensions need not be integers.

The total quantum dimension of a topological theory is then the square-sum of the quantum dimensions of all anyon species in the theory:
\be
\mathcal{D}^2 = \sum_{\textrm{anyons } k} d_k^2 .
\ee

\subsection{Theories with anyonic excitations}
\label{sec:Drinfeld_double}

In this section we will review the theories that appear in this paper. We discuss two Abelian theories: the doubles $\mathcal{D}(\mathbb{Z}_2)$ and $\mathcal{D}(\mathbb{Z}_3)$, as well as two non-Abelian theories: the Chern-Simons theory $SU(2)_4$, and the simplest non-Abelian double $\cD(S_3)$, which is in the focus of this paper.

\subsubsection{$\mathcal{D}(\mathbb{Z}_2)$, the simplest Drinfeld double}
\label{subsubsec:toric_code}

One algebraic model that produces anyons is the Drinfeld double. The Drinfeld double of group $G$, $\mathcal{D}(G)$ is a quasi-triangular Hopf-algebra \cite{Drinfeld86}, with elements in the form of pairs of $g,h \in G$: $(g,h)$, hence the name "double" for this theory. The irreducible representations of this algebra describe anyons, Abelian ones if the group $G$ is Abelian, and non-Abelian anyons if $G$ is non-Abelian.

Anyons in Drinfeld doubles have \emph{flux} and \emph{charge} labels, just as the anyons of Sec.~\ref{sec:anyons}. The flux labels are conjugacy classes of $G$. Charge labels of anyons are either irreps of $G$ (for chargeons), or irreps of normalizer subgroups of $G$ (for dyons).

The simplest Drinfeld double is the toric code, based on $G=\mathbb{Z}_2$: $\mathcal{D}(\mathbb{Z}_2)$; it has Abelian anyons. Let us look at this example.

The group $\mathbb{Z}_2$ has two elements: $\{ e,x \}$, where $e$ is the identity, and $x \cdot x = e$.

Then, the elements of the algebra $\mathcal{D}(\mathbb{Z}_2)$ are:
\bq
&& (e,e) \\
&& (e,x) \\
&& (x,e) \\
&& (x,x) .
\eq

Meanwhile, anyons of this theory are labeled by conjugacy classes and irreps. The conjugacy classes of $\mathbb{Z}_2$ are simply its elements:
\bq
C_e &=& e \\
C_x &=& x ,
\eq
and these are the flux labels of the theory. The normalizers of both of these conjugacy classes/elements are the full group, $\mathbb{Z}_2$. Thus, the irreps labeling the charges are irreps of $\mathbb{Z}_2$: $\Gamma_1^{\mathbb{Z}_2}$ and $\Gamma_{-1}^{\mathbb{Z}_2}$, regardless of the flux content of the anyon. (See the appropriate subtable in Table~\ref{tab:irreps_of_Z3_Z2} for these irreps.)

The traditional electric and magnetic excitations of the toric code are shown in Fig.~\ref{fig:anyon-splitting-diagram_Z2}. These anyons are the simple juxtapositions of a flux and a charge label. The chargeon $\mathbf{e}$ and the fluxon $\mathbf{m}$ inherit their fusion rules from $\mathbb{Z}_2$:
\bq
\mathbf{e} \times \mathbf{e} &=& 1 \\
\mathbf{m} \times \mathbf{m} &=& 1 ,
\eq
and
\be
\mathbf{e} \times \mathbf{m} = \mathbf{em} .
\ee
Similarly, $\mathbf{em} \times \mathbf{e} = \mathbf{m}$, $\mathbf{em} \times \mathbf{m} = \mathbf{e}$, and $\mathbf{em} \times \mathbf{em} = 1$.

\begin{figure}
\begin{centering}
\includegraphics[width=0.15\textwidth]{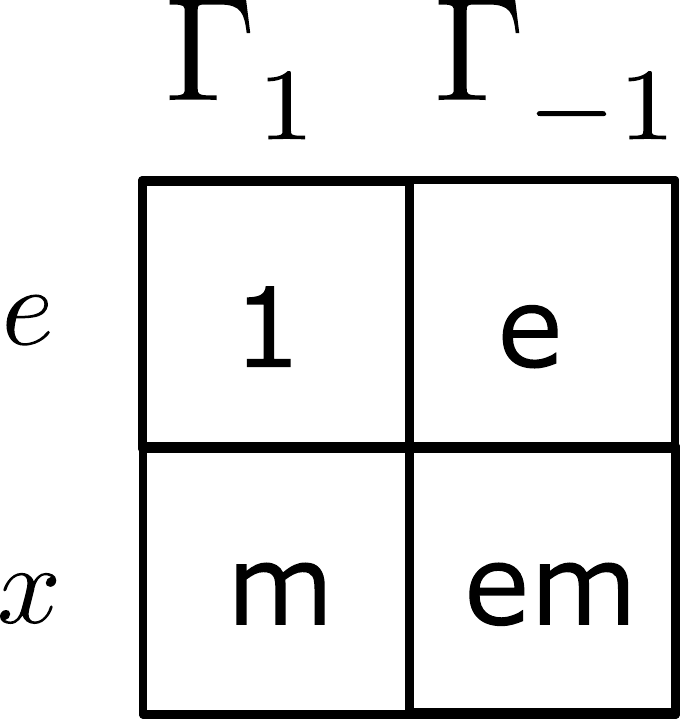}
\caption{Anyons of $\mathcal{D}(\mathbb{Z}_2)$ in relation to their labels.}
\label{fig:anyon-splitting-diagram_Z2}
\end{centering}
\end{figure}

\subsubsection{$\mathcal{D}(\mathbb{Z}_3)$}
\label{subsubsec:D(Z3)}

Another theory with Abelian anyons is $\cD(\mathbb{Z}_3)$, the double of the cyclic group of order $3$. This theory plays a role in Sec.~\ref{sec:phase_diagram}.

The group $\mathbb{Z}_3$ has three elements: $\{e, y, y^2\}$, with $y^3 = e$. The double $\cD(\mathbb{Z}_3)$ then has $9$ distinct elements of its algebra, in the form of $\{ (g,h) | g,h \in \mathbb{Z}_3 \}$.

The flux part of anyons of the model are labeled by conjugacy classes of $\mathbb{Z}_3$:
\bq
C_e &=& e \;\; (\textrm{trivial flux}) \\
C_y &=& y \;\; (m_1) \\
C_{y^2} &=& y^2 \;\; (m_2)
\eq
each corresponding to a flux label from the set $\{ \textrm{triv},m_1,m_2 \}$.

As the normalizer subgroups of each of these conjugacy classes are $\mathbb{Z}_3$ itself, the charge content of anyons are labeled by irreps of $\mathbb{Z}_3$: $\Gamma_1 \;\; (\textrm{trivial flux})$, $\Gamma_{\omega} \;\; (e_1)$ and $\Gamma_{\bar{\omega}} \;\; (e_2)$ (see the appropriate subtable in Table~\ref{tab:irreps_of_Z3_Z2} for these irreps, $\omega$ and $\bar{\omega} \equiv \omega^2$ are the third complex roots of unity). Each of these introduce a charge label, forming the set of charge labels: $\{ \textrm{triv},e_1,e_2 \}$.

All 9 anyons of the theory are the simple combinations of a flux and charge label from these sets: there's the vacuum ($1$, it has a trivial flux and trivial charge label), there are two distinct chargeons ($e_1$ and $e_2$), two distinct fluxons ($m_1$ and $m_2$), and four dyons ($e_1 m_1$, $e_1 m_2$, $e_2 m_1$, $e_2 m_2$).

All anyons of $\mathcal{D}(\mathbb{Z}_3)$ are Abelian, and their fusion rules can be inferred from the fusions of pure chargeons and fluxons. For the chargeons:
\bq
e_1 \times e_1 &=& e_2 \\
e_1 \times e_2 &=& 1 \\
e_2 \times e_2 &=& e_1
\eq
and the rules are identical for the fluxons. For additional details of this model, please refer to Ref.~\cite{Bonderson07}. The $S$-matrix is provided in Appendix~\ref{sec:fusion_Smatrix_DS3} of the current paper.

\subsubsection{The theory $SU(2)_4$}
\label{subsubsec:SU(2)_4}

One theory with non-Abelian excitations is $SU(2)_4$, which will play a role in Sec.~\ref{sec:phase_diagram}.

An $SU(2)_k$ theory is a "q-deformed" version of the $SU(2)$ theory, with $q = \exp(2 \pi i /(k+2))$ ($k$ is an integer) \cite{Bonderson07}. It is also a Chern-Simons field theory with the choice of gauge group $G = SU(2)$ and with a coupling constant $k$ \cite{NSS+08}. These are the archetypical topological quantum field theories, relevant to the fractional quantum Hall states with filling fraction $1/k$ (when $k$ is odd).

The theory $SU(2)_4$ is then an $SU(2)_k$ Chern-Simons theory with coupling $k=4$. Unlike other theories reviewed in thie section, $SU(2)_4$ is chiral. We can gain intuition about its anyons by considering the group $SU(2)$, truncated at spin $k/2 = 2$. Thus the set of spins (and anyons) of the theory:
\be
J_0, \; J_{1/2}, \; J_1, \; J_{3/2}, \; J_2 .
\ee

Some of their fusion rules can be inferred from fusing spins in $SU(2)$, for example:
\bq
J_0 \times J_0 &=& J_0 \\
J_{1/2} \times J_{1/2} &=& J_0 + J_1 \\
J_1 \times J_1 &=& J_0 + J_1 + J_2
\eq
The complete list of fusion rules is shown in Table~\ref{tab:fusion_SU(2)_4}. For additional details of this model, see Refs.~\cite{Bonderson07,NSS+08}.

\begin{table}
	\centering
		\begin{tabular}{ |c||c|c|c|c|c| } 
			\hline
			 & $J_0$ & $J_{1/2}$ & $J_1$ & $J_{3/2}$ & $J_2$ \\
			\hline
			\hline
			$J_0$ & $J_0$ & $J_{1/2}$ & $J_1$ & $J_{3/2}$ & $J_2$ \\ 
			\hline
			$J_{1/2}$ & $J_{1/2}$ & $J_0 \oplus J_1$ & $J_{1/2} \oplus J_{3/2}$ & $J_1 \oplus J_2$ & $J_{3/2}$ \\ 
			\hline
			$J_1$ & $J_1$ & $J_{1/2} \oplus J_{3/2}$ & $J_0 \oplus J_1 \oplus J_2$ & $J_{1/2} \oplus J_{3/2}$ & $J_1$ \\
			\hline
			$J_{3/2}$ & $J_{3/2}$ & $J_1 \oplus J_2$ & $J_{1/2} \oplus J_{3/2}$ & $J_0 \oplus J_1$ & $J_{1/2}$ \\ 
			\hline
			$J_2$ & $J_2$ & $J_{3/2}$ & $J_1$ & $J_{1/2}$ & $J_0$ \\
			\hline
		\end{tabular}
		\caption{Fusion rules of anyons of $SU(2)_4$.}
		\label{tab:fusion_SU(2)_4}
\end{table}

\subsubsection{$\mathcal{D}(S_3)$}

Another theory exhibiting non-Abelian anyons is $\mathcal{D}(S_3)$, the simplest non-Abelian double. It is based on the smallest non-Abelian group, the symmetry group of order 3: $S_3$. This group is isomorphic to the symmetry transformations of an equilateral triangle. Its elements are easily enumerated by thinking about the symmetries of this triangle (see Fig.~\ref{fig:triangle_symmetries}):
\bq
e	\;\; && \textrm{identity} \\
\{ x, xy, xy^2 \} \;\; && \textrm{mirrorings to 3 axes} \\
\{ y, y^2 \} \;\; && \textrm{rotations by $\pi/3$ and $2 \pi/3$} 
\eq

\begin{figure}
\begin{centering}
\includegraphics[width=0.2\textwidth]{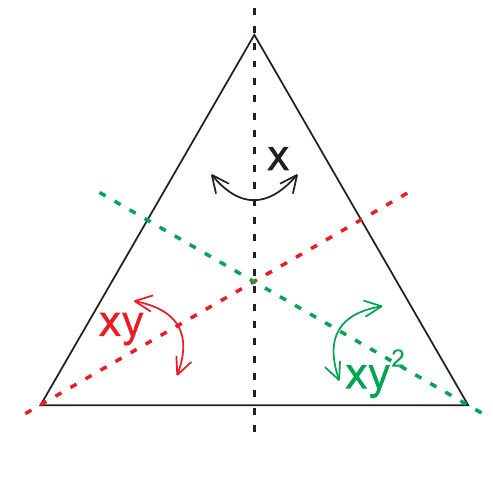} \hspace{1cm}
\includegraphics[width=0.2\textwidth]{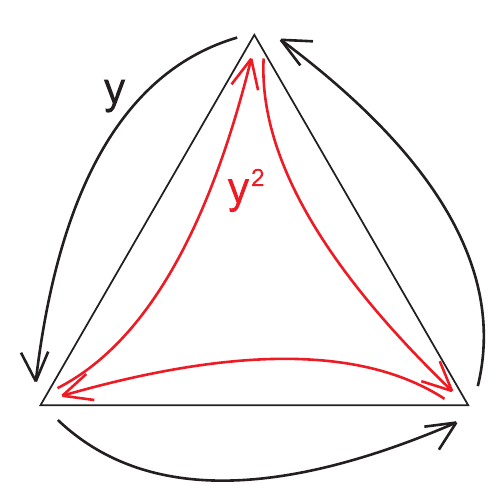}
\caption[Elements of the group $S_3$ as symmetry transformations of an equilateral triangle]{Elements of the group $S_3$ as symmetry transformations of an equilateral triangle: (a) mirrorings $x$, $xy$, $xy^2$, (b) rotations $y$, $y^2$, and the identity transformation (not shown).}
\label{fig:triangle_symmetries}
\end{centering}
\end{figure}

Considering the double $\mathcal{D}(S_3)$, one can then dutifully enumerate all the elements of the algebra: they will be in the form of $\{ (g,h) | g,h \in S_3 \}$, and there will be a total of $36$ of them. However, the anyons of this model, similarly to those of $\mathcal{D}(\mathbb{Z}_2)$ and $\mathcal{D}(\mathbb{Z}_3)$ will be labeled by conjugacy classes and irreps of normalizer subgroups of $S_3$.

Fluxes of the double are labeled by the conjugacy classes of $S_3$:
\bq \label{eq:conj_classes_1}
C_e &=& \{ e \} \\
C_x &=& \{ x,xy,xy^2 \} \\
\label{eq:conj_classes_3}
C_y &=& \{ y,y^2 \}
\eq

The charge labels of $\mathcal{D}(S_3)$ will be irreps of normalizers. The normalizers of $S_3$ are
\bq
\mathcal{N}_e &=& S_3 \\
\mathcal{N}_x &=& \{ e,x \} \cong \cN_{xy} \cong \cN_{xy^2}\cong \mathbb{Z}_2  \\
\mathcal{N}_y = \cN_{y^2} &=& \{ e,y,y^2 \} \cong \mathbb{Z}_3 
\eq

Notice that sometimes normalizers of different elements of a conjugacy class differ (e.g. for elements $x$, $xy$ and $xy^2$ above), however, as remarked in Sec.~\ref{sec:anyons}, these are isomorphic to each other. Thus, the irreps on them will be identical.

The irreps of each of these normalizer subgroups are listed in Tables~\ref{tab:irreps_of_S3}-\ref{tab:irreps_of_Z3_Z2}; the irreps of $\cN_{C_g}$ are the possible charge labels for a dyon with flux label $C_g$. In those tables, and for the remainder of this paper, $\omega = \exp(2 \pi i / 3)$ and $\bar{\omega} = \omega^2 = \exp(4 \pi i / 3)$ are the third complex roots of unity.

\begin{table}
	\centering
		\begin{tabular}{|c||cccccc|}
			\hline
			$S_3$ & $e$ & $y$ & $y^2$ & $x$ & $xy$ & $xy^2$ \\
			\hline
			\hline
			$\Gamma_1^{S_3}$ & 1 & 1 & 1 & 1 & 1 & 1 \\
			\hline
			$\Gamma_{-1}^{S_3}$ & 1 & 1 & 1 & -1 & -1 & -1 \\
			\hline
			$\Gamma_{2}^{S_3}$ & $\begin{pmatrix} 1 & 0 \\ 0 & 1 \end{pmatrix}$ 
			& $\begin{pmatrix} \bar{\omega} & 0 \\ 0 & \omega \end{pmatrix}$ 
			& $\begin{pmatrix} \omega & 0 \\ 0 & \bar{\omega} \end{pmatrix}$ 
			& $\begin{pmatrix} 0 & 1 \\ 1 & 0\end{pmatrix}$ 
			& $\begin{pmatrix} 0 & \omega \\ \bar{\omega} & 0\end{pmatrix}$ 
			& $\begin{pmatrix} 0 & \bar{\omega} \\ \omega & 0\end{pmatrix}$ \\
			\hline
		\end{tabular}
		\caption{Irreducible representations of $\cN_{C_e} = S_3$.}
		\label{tab:irreps_of_S3}
\end{table}

\begin{table}
	\centering
	\begin{tabular}{|c||ccc|}
			\hline
			$\mathbb{Z}_3$ & $e$ & $y$ & $y^2$ \\
			\hline
			\hline
			$\Gamma_1^{\mathbb{Z}_3}$ & 1 & 1 & 1 \\
			\hline
			$\Gamma_{\omega}^{\mathbb{Z}_3}$ & 1 & $\omega$ & $\bar{\omega}$ \\
			\hline
			$\Gamma_{\bar{\omega}}^{\mathbb{Z}_3}$ & 1 & $\bar{\omega}$ & $\omega$ \\
			\hline
		\end{tabular}
		\hspace{1cm}
		\begin{tabular}{|c||cc|}
			\hline
			$\mathbb{Z}_2$ & $e$ & $x$ \\
			\hline
			\hline
			$\Gamma_1^{\mathbb{Z}_2}$ & 1 & 1 \\
			\hline
			$\Gamma_{-1}^{\mathbb{Z}_2}$ & 1 & -1 \\
			\hline
		\end{tabular}
		\caption{Irreducible representations of (a) $\cN_{C_y} = \mathbb{Z}_3$ and (b) $\cN_{C_x} = \mathbb{Z}_2$.}
		\label{tab:irreps_of_Z3_Z2}
\end{table}

The full list of anyons of $\mathcal{D}(S_3)$ is given in Table~\ref{tab:D(S3)_anyons}.

\begin{table}
	\centering
		\begin{tabular}{|c||ccccc|}
			\hline
			Label & $C_g$ & $\mathcal{N}_g$ & Irrep. & Q.dim. & Type \\
			\hline
			\hline
			A & $C_e$ & $S_3$ & $\Gamma^{S_3}_1$ & 1 & vacuum \\
			\hline
			B & $C_e$ & $S_3$ & $\Gamma^{S_3}_{-1}$ & 1 & chargeon \\
			\hline
			C & $C_e$ & $S_3$ & $\Gamma^{S_3}_2$ & 2 & chargeon \\
			\hline
			\hline
			D & $C_x$ & $\mathbb{Z}_2$ & $\Gamma^{\mathbb{Z}_2}_1$ & 3 & fluxon \\
			\hline
			E & $C_x$ & $\mathbb{Z}_2$ & $\Gamma^{\mathbb{Z}_2}_{-1}$ & 3 & dyon \\
			\hline
			\hline
			F & $C_y$ & $\mathbb{Z}_3$ & $\Gamma^{\mathbb{Z}_3}_1$ & 2 & fluxon \\
			\hline
			G & $C_y$ & $\mathbb{Z}_3$ & $\Gamma^{\mathbb{Z}_3}_{\omega}$ & 2 & dyon \\
			\hline
			H & $C_y$ & $\mathbb{Z}_3$ & $\Gamma^{\mathbb{Z}_3}_{\bar{\omega}}$ & 2 & dyon \\
			\hline
		\end{tabular}
		\caption{Anyons of $\mathcal{D}(S_3)$ with their charge and flux labels, quantum dimensions and type.}
		\label{tab:D(S3)_anyons}
\end{table}

Other properties of $\mathcal{D}(S_3)$, such as the complete set of fusion rules of its anyons, and its $S$-matrix can be found in Appendix~\ref{sec:fusion_Smatrix_DS3}.

\subsubsection{Relation between anyons}
\label{subsubsec:irrep-splitting}

One peculiar mathematical property of Drinfeld doubles $\mathcal{D}(G)$ is the relation of irreps of $G$ to the irreps of its normalizer subgroups. It turns out that while gauge-invariant charge labels of dyons are indeed irreps of normalizers, these can be connected to irreps of the full group, $G$. For example, for $\mathcal{D}(S_3)$ we find that the various charge labels are related to each other through restriction of irreps. If we restrict irreps of $S_3$ to the normalizer group $\cN_x$:
\bq \label{eq:irrep-splitting_Nx_1}
\Gamma_1^{S_3} |_{\cN_x} &=& \Gamma_1^{\cN_x} \\
\Gamma_{-1}^{S_3} |_{\cN_x} &=& \Gamma_{-1}^{\cN_x} \\
\label{eq:irrep-splitting_Nx_3}
\Gamma_{2}^{S_3} |_{\cN_x} &=& \Gamma_1^{\cN_x} \oplus \Gamma_{-1}^{\cN_x} 
\eq
Similarly, restricting irreps of $S_3$ to the group $\cN_y$:
\bq \label{eq:irrep-splitting_Ny_1}
\Gamma_1^{S_3} |_{\cN_y} &=& \Gamma_1^{\cN_y} \\
\label{eq:irrep-splitting_Ny_2}
\Gamma_{-1}^{S_3} |_{\cN_y} &=& \Gamma_1^{\cN_y} \\
\label{eq:irrep-splitting_Ny_3}
\Gamma_{2}^{S_3} |_{\cN_y} &=& \Gamma_\omega^{\cN_y} \oplus \Gamma_{\bar{\omega}}^{\cN_y} 
\eq

The diagram in Fig.~\ref{fig:orth_resolution}, which is a generalized, non-Abelian version of Fig.~\ref{fig:anyon-splitting-diagram_Z2}, now for $\mathcal{D}(S_3),$ summarizes these connections. Rows of this diagram correspond to flux labels (conjugacy classes) of the theory, and in the rows anyons are shown that are labeled by those flux labels. Notice that there are six rows of the diagram, while only three flux labels: each flux label incorporates as many rows as is the cardinality of the appropriate conjugacy class; and we may think about each individual row as labeled by a specific element of the conjugacy class.

Columns of the diagram are labeled by irreps of the group $S_3$: the three irreps $\Gamma_1^{S_3}$, $\Gamma_{-1}^{S_3}$ and $\Gamma_2^{S_3}$ each label as many columns as is the number of independent parameters for that irrep ($1$ and $1$ for the one-dimensional irreps, and $4$ for the two-dimensional irrep). Anyons are distributed amongst these charge labels according to their charge labels. An anyon label is shown in a certain column if and only if that anyon has \emph{charge flavors} described by that charge label. We can deduce these relations by taking the irrep splitting relations Eqs.~\eqref{eq:irrep-splitting_Nx_1}-\eqref{eq:irrep-splitting_Ny_3} into account.

The connection between anyon labels and irreps of $S_3$ detailed here will have an interesting consequence for any analysis conducted for this double. If we were to fundamentally modify all elements of the algebra related to an irrep of $S_3$, several anyons and anyon flavors would be affected by that: all of those related to the modified charge label through these intricate irrep-restrictions. As an example, modifying (e.g. completely removing) the charge label $\Gamma^{S_3}_{-1}$ would not only have an effect on anyon $B$ (the one directly labeled by $\Gamma^{S_3}_{-1}$), but also on anyons $E$ and $F$, as some of their charge flavors are related to the modified irrep $\Gamma_{-1}^{S_3}$ (see Fig.~\ref{fig:orth_resolution}).

An additional useful property of the diagram in Fig.~\ref{fig:orth_resolution} is that the number of distinct squares corresponding to an anyon equals its squared quantum dimension \cite{KL17}, each distinct square of the diagram representing a different flavor of that anyon. Meanwhile, the total number of squares in the diagram ($36$) is exactly the square dimension of the total quantum dimension of this theory. These properties will be utilized in our analysis in Sec.~\ref{sec:phase_diagram} to predict dimensions of new anyons, obtained by forbidding some of their flavors through the forbiddance of some of the charge and/or flux labels.

\begin{figure}
\begin{centering}
\includegraphics[width=0.4\textwidth]{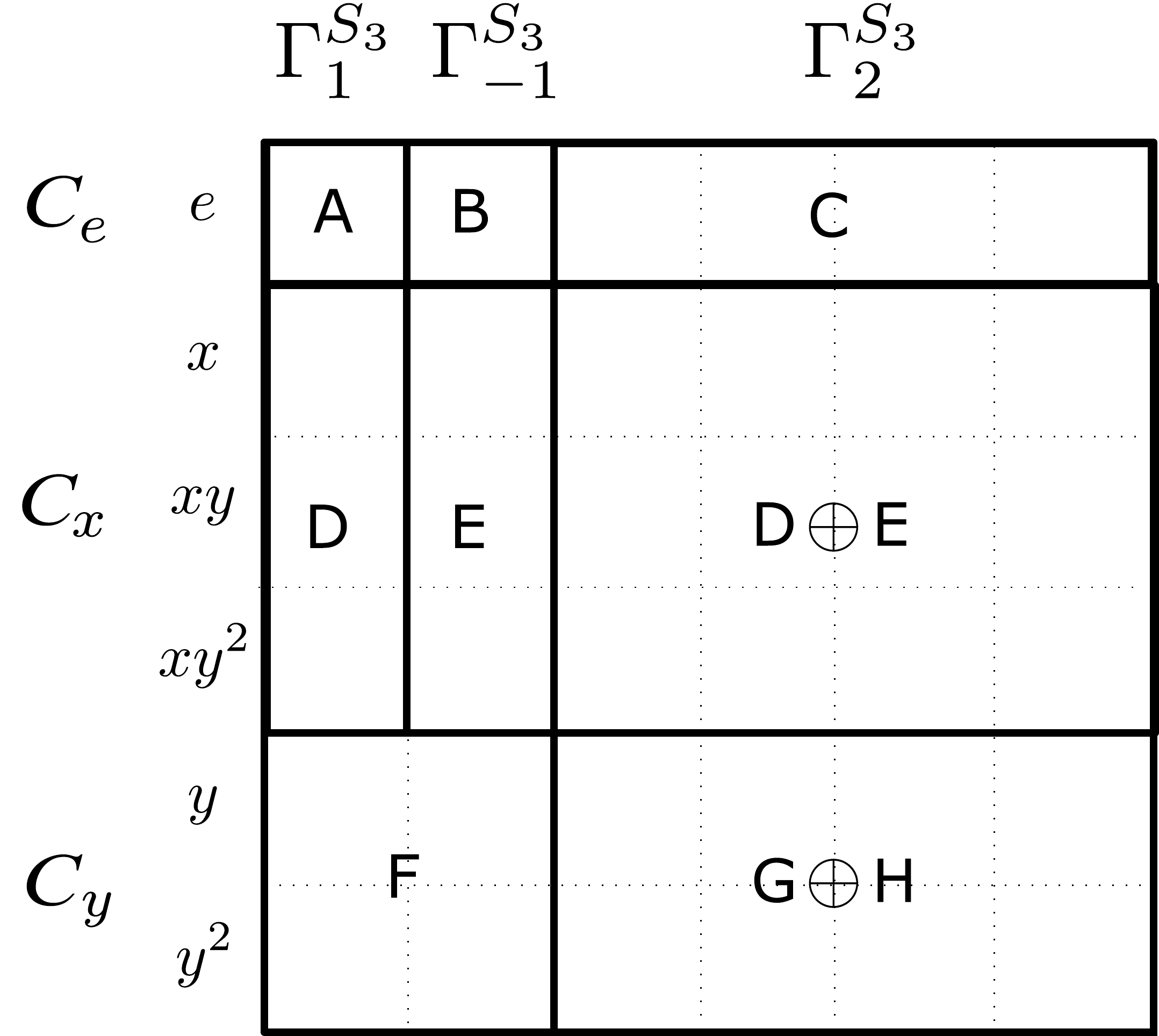}
\caption[Connections between the sets of labels $\{ \Gamma^{S_3}_1, \Gamma^{S_3}_{-1}, \Gamma^{S_3}_2 \}$, $\{ C_e, C_x, C_y \}$ and (flavors of) anyons of $\mathcal{D}(S_3)$]{Connections between the sets of labels $\{ \Gamma^{S_3}_1, \Gamma^{S_3}_{-1}, \Gamma^{S_3}_2 \}$, $\{ C_e, C_x, C_y \}$ and (flavors of) anyons of $\mathcal{D}(S_3)$. The number of rows and columns for each flux and charge label correspond to the number of distinct flavors that label describes. The number of squares corresponding to an anyon label is the squared dimension of that anyon.}
\label{fig:orth_resolution}
\end{centering}
\end{figure}

As a final note, we'd like to point out that such connections between irreps of the full group $G$ and irreps of the normalizer subgroups can always be established, for any group $G$, and a diagram equivalent to Fig.~\ref{fig:orth_resolution} can always be drawn \cite{KL17}. For Abelian doubles ($G = \mathbb{Z}_d$), these connections will be trivial. The normalizer subgroups in such a case will be the full group, $G$, thus the irreps of $G$ and those of the normalizers will be exactly the same. For non-Abelian doubles, the connections will be more interesting, as there will be non-trivial normalizer subgroups, whose irreps will sometimes split (Eqs.~\eqref{eq:irrep-splitting_Ny_1}-\eqref{eq:irrep-splitting_Ny_2} above), sometimes merge (Eq.~\eqref{eq:irrep-splitting_Nx_3} or Eq.~\eqref{eq:irrep-splitting_Ny_3}) to yield an irrep of $G$.

\section{Phase transitions in Drinfeld doubles}
\label{sec:phase_transitions_general}

In this paper we analyze phase transitions in Drinfeld doubles, induced by \emph{forbidding} one (or several) flux or charge labels. Anyons described by the forbidden label cannot be created, both as thermal excitations from the vacuum, and through braiding/fusions of other anyons. Forbidding anyons in the double $\mathcal{D}(G)$ will modify the dynamics and result in the emergence of a new theory for the unforbidden set of anyons.

In Sec.~\ref{subsec:mechanisms} we review the physical processes the double $\mathcal{D}(G)$ undergoes through these phase transitions. Then, in Sec.~\ref{subsec:protocol} we present the mathematical procedure for modeling the transitions, and how one can find the proper emergent theory in all cases. We will present our detailed analysis of the phase transitions in $\cD(S_3)$ in Sec.~\ref{sec:phase_diagram} and in Appendix~\ref{sec:all_theories}.

\subsection{Mechanisms of phase transition}
\label{subsec:mechanisms}

To find the new topological theory emerging after forbidding flux and charge labels, one can picture that the original theory, i.e, $\mathcal{D}(G)$, undergoes a series of physical processes. These processes are conceptual tools to illustrate the way the topological model restructures itself. This representation of our results is complementary to mathematical transformations on the $S$-matrix, a procedure described in Sec.~\ref{subsec:protocol}.

Those processes, which we will review in details in the following paragraphs are 
\begin{enumerate}
 \item Forbiddance of certain anyons, related to the forbidden label;
 \item Condensation of some particles, i.e., they become indistinguishable from the vacuum;
 \item Non-Abelian anyons with quantum dimension larger than one splitting up into lower-dimensional anyons, due to their previously indistinguishable flavors becoming distinguishable in the new theory;
 \item (Additional) partial forbiddance/spontaneous symmetry breaking of anyons.
\end{enumerate}

\subsubsection{Forbiddance}

Anyons that are never created in the new theory, because of forbidding a label (a conjugacy class or an irrep of $G$) they are related to, become \emph{forbidden}. This is a fundamental process that happens in every one of the phase transitions we consider. In fact, forbiddance of these particles is what induces the other processes listed in the current section.

Forbidden anyons are neither created through fusion, nor from the vacuum. These constraints have a practical significance for our calculations, as we will use these properties to find the new theories. In particular, our findings will be justified by performing transformations on the S-matrix, a mathematical procedure described in Sec.~\ref{subsec:protocol}.

\subsubsection{Condensation}

In certain cases the resulting set of anyons after forbiddance will not form a valid theory. In these cases, however, some particles become indistinguishable from the vacuum, and \emph{condense} to it.

This happens when through forbiddance we remove the full set of anyons that have made the now indistinguishable anyons previously distinguishable in the original theory. For example, anyons $A$ (the vacuum) and $B$ of $\cD(S_3)$ are only distinguished through their braiding relations to anyons labeled by the conjugacy class $C_x$ (anyons $D$ and $E$). Thus, when we forbid $C_x$ in Sec.~\ref{subsubsec:DZ3}, anyon $B$ condenses to the vacuum, $A$.

We identify these processes through the $S$-matrix of the new theory: when the new $S$-matrix has identical (or linearly dependent) entries for two particles, that is an indicator that those two particles have become indistinguishable from each other in the new theory.

\subsubsection{Splitting of particles}

Anyons with dimension higher than $1$ in the original model sometimes \emph{split up} into two or more 1-dimensional anyons in the new theory. This splitting happens when the different flavors of the anyon in question, indistinguishable in the original model by braiding or fusion, become distinguishable in the new theory. Thus, the newly distinguishable states form separate anyons.

Mathematically, this process can be shown by relating the $S$-matrix of (the new) \emph{theory A} to the $S$-matrix of (a different) \emph{theory B}. If by merging anyons $g$ and $h$ in \emph{theory B} we arrive to the $S$-matrix of \emph{theory A}, with merged anyon $gh$, we conclude that the $gh$ anyon of \emph{theory A} could split up and become the anyons $g$ and $h$ in \emph{theory B}.

An example of this is analyzed in Sec.~\ref{subsubsec:DZ3}. After anyon $B$ condenses to the vacuum, we arrive at an $S$-matrix that is identical to the matrix we get when we merge fluxes and charges in $\mathcal{D}(\mathbb{Z}_3)$. Thus, the emergent theory of Sec.~\ref{subsubsec:DZ3}, and the quantum double $\mathcal{D}(\mathbb{Z}_3)$ are related through a series of condensation and splitting of anyons.

\subsubsection{Partial forbiddance, Spontaneous symmetry breaking}

\emph{Partial forbiddance} of an anyon happens when flavors otherwise indistinguishable split up in a way that some of the flavors become completely forbidden, while others remain unaffected.

When partial forbiddance emerges as a process induced by previous forbiddance and condensation processes, the accompanying symmetry breaking is spontaneous. Alternatively, when partial forbiddance is forced unto the model, the symmetry breaking is not spontaneous. Examples of both versions can be seen in Secs.~\ref{subsec:su2_4_1}-\ref{subsec:su2_4_2}.

Both versions of partial forbiddance will lead to the anyons in question losing some of their flavors, and as a result their quantum dimensions will decrease accordingly.

\subsection{Our protocol}
\label{subsec:protocol}

The protocol to find the new theories emerging from $\mathcal{D}(G)$ is the following. First, we forbid a (set of) conjugacy classes and irreps, i.e. all anyons and anyon flavors related to those labels. They won't be created either from the vacuum or through fusion processes. Then, we will construct the $S$-matrix of the new theory, based on the new set of fusion rules for the subset of unforbidden anyons.

\subsubsection{Truncating fusion rules}
To enforce forbiddance of anyons, we start with the original fusion rules of $\mathcal{D}(G)$ and truncate them, so that the forbidden anyons are never created.

\subsubsection{Constructing the $N$-symbols}
We use the remaining (now closed) set of fusion rules to construct the $N$-symbols (see Eq.~(\ref{eq:N_symbol_definition})) for each particle. In this step we can utilize an anyon-diagram of $\mathcal{D}(G)$, a generalized version of Figs.~\ref{fig:anyon-splitting-diagram_Z2} and \ref{fig:orth_resolution}, that encapsulates the relation between irreps of $G$ and irreps of its normalizer subgroups. We are guided by the assumption that the quantum dimensions of anyons are those inferred from this diagram, i.e. the quantum dimension changes according to the diagram when the anyon in question is partially forbidden, otherwise it is unchanged from the dimensions of $\mathcal{D}(G)$. However, in certain cases assuming the quantum dimensions this way doesn't yield a valid theory, particularly when a spontaneous symmetry breaking is among the physical processes. In those cases we aim to find a theory with the appropriate fusion rules only, and disregard (some of) the values of quantum dimensions.

\subsubsection{Inverting the Verlinde-formula}
We then compute the eigenvalues of the $N$-symbols, and use the relation between these eigenvalues and the elements of the $S$-matrix \cite{Kitaev06} to invert the (\ref{eq:verlinde}) Verlinde-formula and reverse-engineer the $S$-matrix. Throughout this process, we use the symmetries of the $S$-matrix of the original $\mathcal{D}(G)$ model, only to make a choice between new $S$-matrices that otherwise equivalently reproduce the fusion rules and quantum dimensions of the new model. We find the prefactor (or the exact entries, rather than simply their relative values) of the $S$-matrix by enforcing that it be unitary.

If we arrive at an $S$-matrix that has linearly dependent entries for certain anyons (it does not recreate all the fusion rules we have started with, and is not unitary in these cases), we conclude that further physical processes will happen: condensation of some particles. For a list of the possible physical processes, see Sec.~\ref{subsec:mechanisms}.

\section{Phases of $\mathcal{D}(S_{3})$}
\label{sec:phase_diagram}

In this section we will present phases that can emerge from the theory $\mathcal{D}(S_3)$ through a series of physical processes, outlined in Sec.~\ref{sec:phase_transitions_general}. We focus on $\mathcal{D}(S_3)$ as it is the simplest non-Abelian double, and as such novel theories could emerge from it.

We consider all combinations of charge and flux labels of $\mathcal{D}(S_3)$, and derive the theory that emerges as a result of forbidding each set of labels. The exact protocol to find the emerging theories is outlined in Sec.~\ref{subsec:protocol}. We include three representative phases in this section, as well as the detailed analysis of how those theories emerge from $\mathcal{D}(S_3)$. The complete description of all possible emergent theories is in Appendix~\ref{sec:all_theories}.

\subsection{Forbidding the conjugacy class $C_x$ leads to $\cD(\mathbb{Z}_3)$} \label{subsubsec:DZ3}

Let us first give a high-level, intuitive overview of what happens when we forbid the conjugacy class $C_x$.

First, simply concentrate on the group $S_3 = \{ e,x,xy,xy^2,y,y^2 \}$ instead of the double $\cD(S_3)$. Then, remove the elements $C_x = \{ x,xy,xy^2 \}$ from this group. The remaining set of elements is $\{ e,y,y^2 \}$, which themselves form a group: $\mathbb{Z}_3$.

It turns out that when we do the same procedure for $\cD(S_3)$, and forbid the conjugacy class $C_x$ (i.e. the elements $\{ x,xy,xy^2 \}$ among the flux labels), the remaining set of anyons will form the theory $\cD(\mathbb{Z}_3)$.

However, this parallel is not at all trivial, as we have never considered the charge labels in this argument. What happens is the following process.

\textit{Step 1 ---} Forbidding the conjugacy class $C_x$ confines anyons $D$ and $E$ of the original model, leaving the set $\{A, B, C, F, G, H\}$.

\textit{Step 2 ---} As anyons $A$ and $B$ only differ through their braiding relations with, the now forbidden, anyons $D,E$ (see Appendix~\ref{sec:fusion_Smatrix_DS3}), anyons $A$ and $B$ will become indistinguishable in the new theory. $B$ will condense to the vacuum, $A$.

This induces the partial forbiddance of $C,F,G,H$, which all lose their antisymmetric charge subspaces (as the fundamental antisymmetric chargeon, $B$ disappeared).

Notice, that this is the step where all antisymmetric charge flavors disappear, necessary for the eventual emergence of $\cD(\mathbb{Z}_3)$. (Just as removing the $C_x$ labels from the flux flavors was necessary.)

\textit{Step 3 ---} This is followed by the split-up of higher-dimensional anyons $C, F, G, H$ into two distinct 1-dimensional, Abelian anyons each. This split-up can happen, as the fluxons previously making their flavors indistinguishable were part of the $C_x$ conjugacy class, now forbidden.

For example, anyon $F$ originally had two flux flavors: $y$ and $y^2$, being indistinguishable due to the braiding process $x y x = y^2$, which would transform one flavor to the other. After forbidding the class $C_x$, such a process is now forbidden, and the $y$ and $y^2$ flavors of $F$ become distinguishable.

The final set of anyons is thus: $\{A',C_a,C_b,F_a,F_b,G_a,G_b,H_a,H_b\}$, which will form the quantum double $\mathcal{D}(\mathbb{Z}_3)$. Labels in the form of $X_a$ and $X_b$ correspond to two distinct one-dimensional flavors of anyon $X$, which have split up to form two separate anyons. The correspondence between standard labels of $\mathcal{D}(\mathbb{Z}_3)$ and our labels is shown in Table~\ref{tab:D(Z3)_labels}. The process described above is shown in Fig.~\ref{fig:Cx_limit}.

\begin{table}
	\centering
		\begin{tabular}{ |c||c|c|c| } 
			\hline
			 & $1$ & $e_1$ & $e_2$ \\
			\hline
			\hline
			$1$ & $A'$ & $C_a$ & $C_b$ \\ 
			\hline
			$m_1$ & $F_a$ & $G_a$ & $H_a$ \\ 
			\hline
			$m_2$ & $F_b$ & $H_b$ & $G_b$ \\
			\hline
		\end{tabular}
		\caption{Correspondence between standard anyon labels of $\mathcal{D}(\mathbb{Z}_3)$ and anyon labels evolved from of $\mathcal{D}(S_3)$, when forbidding the conjugacy class $C_x$.}
		\label{tab:D(Z3)_labels}
\end{table}

\begin{figure*}
\begin{centering}
\includegraphics[width=\textwidth]{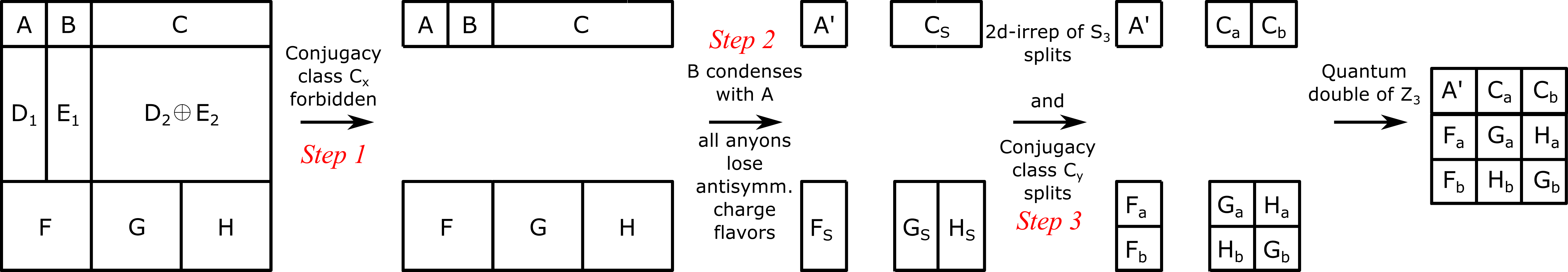}
\caption{Pictorial representation of the process transforming $\mathcal{D}(S_3)$ to $\mathcal{D}(\mathbb{Z}_3)$, when we forbid conjugacy class $C_x$. The subscripts $S$ denote the symmetric subspaces of anyons ($F_S \equiv F_1$ of Fig.~\ref{fig:anyon-splitting-2}), and subscripts $a$ and $b$ denote the final, $1$-dimensional components of the anyons. $A'$ is the vacuum after the condensation of $B$. (Forbidding $C_x$ and $\Gamma_{-1}$  at the same time, as done in Sec.~\ref{subsec:DZ3_2}, will modify \textit{Step 2} of this process: instead of condensation, $B$ will be forbidden, otherwise the rest of the steps are identical.)}
\label{fig:Cx_limit}
\end{centering}
\end{figure*}

\subsubsection*{Proof}

We can prove that anyons of $\mathcal{D}(S_3)$ indeed undergo the above process by following the protocol detailed in Sec.~\ref{subsec:protocol}. Throughout this proof, we will refer to the steps numbered in the overview given above, which coincide with the steps shown in Fig.~\ref{fig:Cx_limit}.

\textit{Step 1 ---} The first step in this physical process is the forbiddance of anyons $D,E$. In our protocol, we truncate the fusion rules of $\mathcal{D}(S_3)$ for the remaining set of particles $\{A, B, C, F, G, H\}$, these truncated rules are shown in Table~\ref{tab:fusion_ABCFGH}.
\begin{table}
	\centering
		\begin{tabular}{ |c||c|c|c|c|c|c| } 
			\hline
			 & A & B & C & F & G & H \\
			\hline
			\hline
			A & A & B & C & F & G & H \\ 
			\hline
			B & B & A & C & F & G & H \\ 
			\hline
			C & C & C & A $\oplus$ B $\oplus$ C & G $\oplus$ H & F $\oplus$ H & F $\oplus$ G \\
			\hline
			F & F & F & G $\oplus$ H & A $\oplus$ B $\oplus$ F & H $\oplus$ C & G $\oplus$ C \\
			\hline
			G & G & G & F $\oplus$ H & H $\oplus$ C & A $\oplus$ B $\oplus$ G & F $\oplus$ C \\
			\hline
			H & H & H & F $\oplus$ G & G $\oplus$ C & F $\oplus$ C & A $\oplus$ B $\oplus$ H \\
			\hline
		\end{tabular}
		\caption{Fusion rules of remaining particles after forbidding conjugacy class $C_x$.}
		\label{tab:fusion_ABCFGH}
\end{table}

Using the new set of fusion rules we can construct the new $S$-matrix, following our protocol, and we have:
\be \label{eq:SmatrixABCFGH}
S=\frac{1}{6}\left[\begin{array}{cccccc}
1 & 1 & 2 & 2 & 2 & 2 \\
1 & 1 & 2 & 2 & 2 & 2 \\
2 & 2 & 4 & -2 & -2 & -2 \\
2 & 2 & -2 & 4 & -2 & -2 \\
2 & 2 & -2 & -2 & 4 & -2 \\
2 & 2 & -2 & -2 & -2 & 4
\end{array} \right] .
\ee

Now, we can state the following two lemmas:

\begin{lem}[Condensation (\textit{Step 2})] \label{lem:condensation_DZ3}
Anyon $B$ condenses to the vacuum, and the following $S$-matrix describes the emerging theory:
\be \label{DZ3_correspondence_mtx}
\frac{1}{3} \left[\begin{array}{ccccc}
1 & \sqrt{2} & \sqrt{2} & \sqrt{2} & \sqrt{2} \\
\sqrt{2} & 2 & -1 & -1 & -1 \\
\sqrt{2} & -1 & 2 & -1 & -1 \\
\sqrt{2} & -1 & -1 & 2 & -1 \\
\sqrt{2} & -1 & -1 & -1 & 2
\end{array} \right] .
\ee
\end{lem}

\begin{proof}
We transform the \eqref{eq:SmatrixABCFGH} $S$-matrix from its current basis
$$\{A,B,C,F,G,H\}$$
to the basis
$$\{ (A+B)/\sqrt{2}, (A-B)/\sqrt{2}, C, F, G, H \}$$
in the same spirit as how Ref.~\cite{WLW+08} relates the toric code to the Ising model. The transformation applied here yields zero entries for the entire row and column corresponding to $(A-B)/\sqrt{2}$. Dropping these we get the matrix \eqref{DZ3_correspondence_mtx}.
\end{proof}

\begin{lem}[Splitting of anyons (\textit{Step 3})] \label{lem:splitting_DZ3}
Merging anyons of $\mathcal{D}(\mathbb{Z}_3)$ (chargeons $\{ e_1,e_2 \}$, fluxons $\{ m_1,m_2 \}$, and dyons $\{ e_1 m_1,e_2 m_2 \}$, $\{ e_1 m_2,e_2 m_1 \}$, pairwise) yields a block-diagonal $S$-matrix, the physical block of which is \eqref{DZ3_correspondence_mtx}.
\end{lem}

\begin{proof}
We transform the matrix \eqref{eq:DZ3_Smatrix} from its current basis
$$\{ 1, e_1, e_2, m_1, m_2, e_1 m_1, e_2 m_1, e_1 m_2, e_2 m_2 \}$$
to a new basis, which is:
$$1/\sqrt{2} \left\{ \sqrt{2}, e_1+e_2, m_1+m_2, e_1 m_1+e_2 m_2, e_1 m_2+e_2 m_1, \right. $$
$$\left. e_1-e_2, m_1-m_2, e_1 m_1-e_2 m_2, e_1 m_2-e_2 m_1 \right\}$$

In this new basis we find the $S$-matrix is block-diagonal, the "upper block" is formed by the first 5 anyons (symmetric block), the "lower block" is formed by the last 4 anyons (antisymmetric block).

We can drop the antisymmetric block, on the basis that it doesn't include the vacuum state, and if we were to add the vacuum to it, we would get a pathological $S$-matrix with all zero entries for the first row and column. This leaves us the symmetric block, whose entries are identical to those of Eq.~\eqref{DZ3_correspondence_mtx}.
\end{proof}

Combining Lemmas~\ref{lem:condensation_DZ3} and \ref{lem:splitting_DZ3} we can conclude that the theory emerging by forbidding the conjugacy class $C_x$ will undergo a condensation (\textit{Step 2}), followed by splitting of anyons (\textit{Step 3}), and form $\mathcal{D}(\mathbb{Z}_3)$.  \qed

\subsection{Forbidding the conjugacy class $C_y$ leads to $SU(2)_{4}$} \label{subsec:su2_4_1}

Just like in Sec.~\ref{subsubsec:DZ3} we may try to understand the effect of forbidding $C_y$ through a high-level argument. As there, let's start by looking at the group $S_3$ only, rather than the double $\cD(S_3)$.

Forbidding the elements $C_y = \{ y,y^2 \}$ in $S_3$ would leave the set of elements $\{ e,x,xy,xy^2 \}$. However, in this set there is nothing to distinguish between the different flavors of $C_x$, that is between $x$, $xy$ and $xy^2$. Therefore, a spontaneous symmetry breaking happens, and the set of allowed elements transitions into $\{ e,x \}$ (or $\{ e,xy \}$ or $\{ e,xy^2 \}$, all are equivalent). The resulting set is isomorphic to the group $\mathbb{Z}_2$.

Based on this argument, one might argue that forbidding $C_y$ in the quantum double $\cD(S_3)$ will transform it into the quantum double of $\mathbb{Z}_2$: $\cD(\mathbb{Z}_2)$. However, this would be incorrect. Notice that the above argument only concerns the flux labels of the double, and doesn't address the resulting effect on the charge flavors. In fact, the charge flavors remain unaffected! (Note, that this is unlike the case investigated in Sec.~\ref{subsubsec:DZ3}. While there the charge flavors weren't immediately affected either, condensation was induced and that eventually modified the allowed set of charge flavors.)

As the forbiddance of $C_y$ only concerns the flux labels, the resulting theory won't be the simple $\cD(\mathbb{Z}_2)$, rather it could be a non-trivial theory. It turns out to be $SU(2)_4$.

These are the steps that happen (and while the above argument didn't yield the correct theory, it is still helpful to consult the steps the group $S_3$ underwent after forbidding $C_y$):

\textit{Step 1 ---} The process starts with the forbiddance of anyons $F,G,H$.

\textit{Step 2 ---} This in turn induces the partial forbiddance (or spontaneous symmetry breaking) of anyons $D$ and $E$. This happens because by forbidding the full conjugacy class $C_y$ we forbade the flux labels $y,y^2$, the exact fluxes that could transform one flux flavor of $D$ or $E$ (e.g. flavor $x$) into another flux flavor (to $xy$ or $xy^2$), through braiding. Thus, only one flux flavor of them will remain, be it $x$, $xy$ or $xy^2$ (hence the \emph{spontaneous} symmetry breaking).

At the end, we are left with the set of anyons $\{ A,B,C,D,E \}$, where $D$ and $E$ have modified dimensions compared to the original $\mathcal{D}(S_3)$ theory: we predict this modified dimension to be $\sqrt{3}$ for both $D$ and $E$, based on the pictorial argument shown in Fig.~\ref{fig:Cy_limit}.

This set corresponds to $SU(2)_4$ through:
$$A \equiv J_0, \; B \equiv J_2, \; C \equiv J_1, \; D \equiv J_{1/2}, \; E \equiv J_{3/2} $$
with the notation for anyons of $SU(2)_4$ explained in Sec.~\ref{subsubsec:SU(2)_4}.

\begin{figure*}
\begin{centering}
\includegraphics[width=\textwidth]{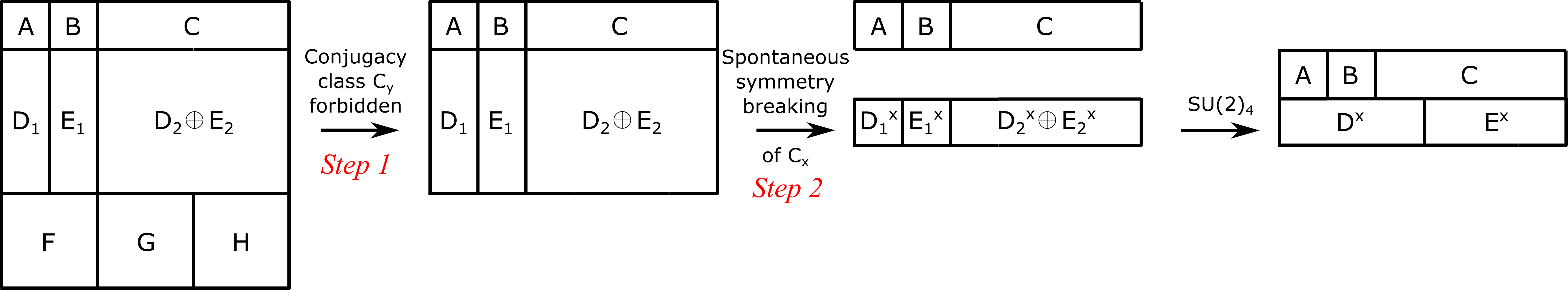}
\caption{Forbidding conjugacy class $C_y$ transforms $\mathcal{D}(S_3)$ to $SU(2)_4$. The superscripts $x$ denote the components of anyons $D$ and $E$ after the symmetry breaking.}
\label{fig:Cy_limit}
\end{centering}
\end{figure*}

\subsubsection*{Proof}

We follow the protocol outlined in Sec.~\ref{subsec:protocol}. We start by truncating the fusion rules of $\mathcal{D}(S_3)$ to the subset of anyons $\{ A,B,C,D,E \}$, these are shown in Table~\ref{tab:fusion_ABCDE}.
\begin{table}
	\centering
		\begin{tabular}{ |c||c|c|c|c|c| } 
			\hline
			 & A & B & C & D & E \\
			\hline
			\hline
			A & A & B & C & D & E \\ 
			\hline
			B & B & A & C & E & D \\ 
			\hline
			C & C & C & A $\oplus$ B $\oplus$ C & D $\oplus$ E & D $\oplus$ E \\
			\hline
			D & D & E & D $\oplus$ E & A $\oplus$ C & B $\oplus$ C \\
			\hline
			E & E & D & D $\oplus$ E & B $\oplus$ C & A $\oplus$ C \\
			\hline
		\end{tabular}
		\caption{Fusion rules of remaining particles after forbidding conjugacy class $C_y$.}
		\label{tab:fusion_ABCDE}
\end{table}

We use this new set of fusion rules to reverse-engineer the corresponding $S$-matrix, as explained in Sec.~\ref{subsec:protocol}. This yields the matrix
\be
S= \frac{1}{\sqrt{12}}\left[\begin{array}{ccccc}
1 & 1 & 2 & \sqrt{3} & \sqrt{3} \\
1 & 1 & 2 & -\sqrt{3} & -\sqrt{3} \\
2 & 2 & -2 & 0 & 0 \\
\sqrt{3} & -\sqrt{3} & 0 & \sqrt{3} & -\sqrt{3} \\
\sqrt{3} & -\sqrt{3} & 0 & -\sqrt{3} & \sqrt{3} 
\end{array} \right] ,
\ee
which is identical to the $S$-matrix of $SU(2)_4$ \cite{Bonderson07}. Further, the quantum dimensions in this new theory match the original quantum dimensions for $A,B,C$, and agree with our prediction for the new, decreased dimensions of $D$ and $E$. \qed

\subsection{Forbidding the irrep $\Gamma_2$ leads to $SU(2)_4$} \label{subsec:su2_4_2}

Forbidding the 2-dimensional irrep $\Gamma_2$ will induce a similar process to that detailed in Sec.~\ref{subsec:su2_4_1}. The only difference is that the two steps detailed there have merged into a single step here:

\textit{Step 1 ---} The forbiddance of $\Gamma_2$ immediately forbids anyons $C,G,H$. Furthermore, it immediately forbids parts of $D$ and $E$, their flavors related to the 2-dimensional irrep.

The symmetry breaking/partial forbiddance of $D,E$ is similar to the one that was induced in Sec.~\ref{subsec:su2_4_1}, to the extent that the quantum dimensions of the forbidden spaces are the same in the two situations. Here, however, the partial forbiddance is not spontaneous: the forbidden flavors are those related to the $\Gamma_2$ irrep.

The set of remaining anyons is then $\{ A,B,D,E,F \}$, with $D,E$ both having modified dimensions of $\sqrt{3}$ once again. Notice, that with the exchange of $F$ to $C$ it is the same set of anyons as in Sec.~\ref{subsec:su2_4_1}. The process described here is shown in Fig.~\ref{fig:Gamma2_limit}.

The set $\{ A,B,D,E,F \}$ corresponds to $SU(2)_4$ through:
$$A \equiv J_0, \; B \equiv J_2, \; D \equiv J_{1/2}, \; E \equiv J_{3/2}, \; F \equiv J_1 $$
for the notation of anyons of $SU(2)_4$, see Sec.~\ref{subsubsec:SU(2)_4}.

\begin{figure}
\begin{centering}
\includegraphics[width=0.5\textwidth]{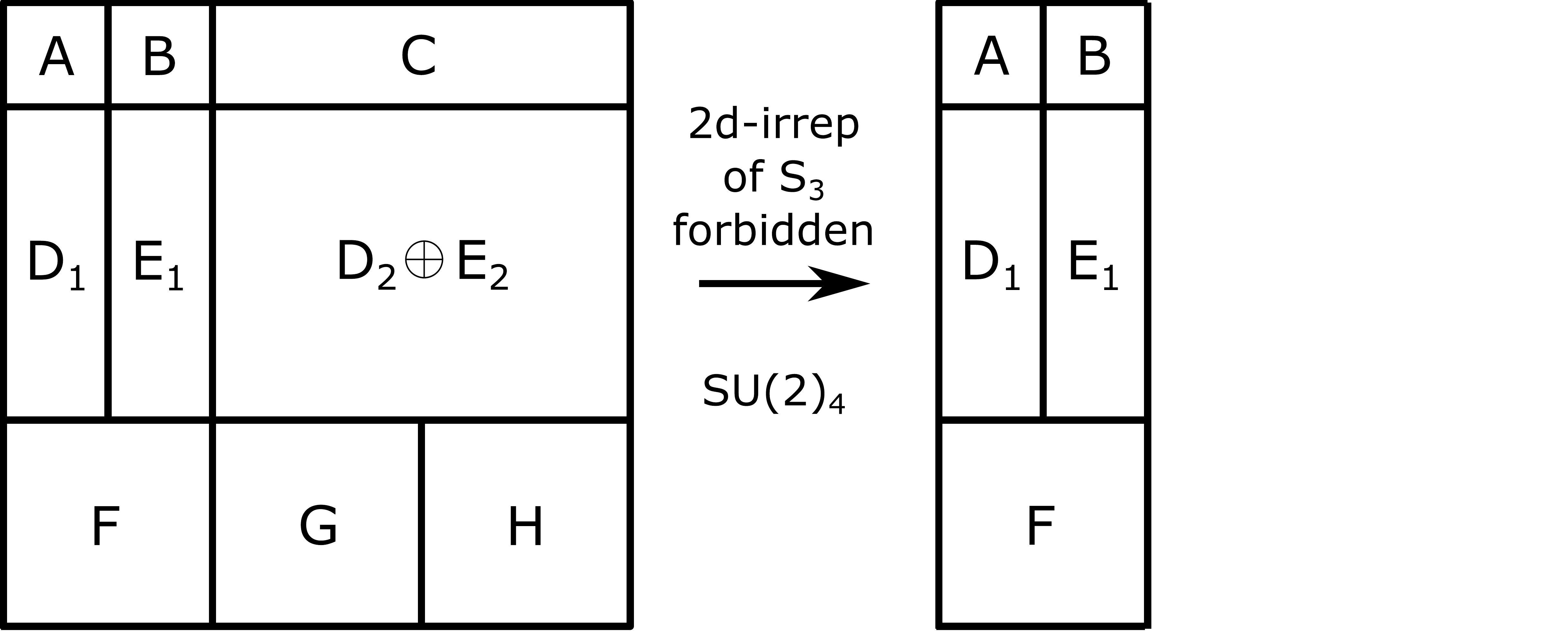}
\caption{Forbidding irrep $\Gamma_2$ of $S_3$ transforms $\mathcal{D}(S_3)$ to $SU(2)_4$.}
\label{fig:Gamma2_limit}
\end{centering}
\end{figure}

\subsubsection*{Proof}

The proof is completely identical to the proof in Sec.~\ref{subsec:su2_4_1}. Simply start by truncating the fusion rules of $\cD(S_3)$ to the current subset of anyons $\{ A,B,D,E,F \}$, then follow the same steps. \qed

\section{Lattice realization}
\label{sec:lattice_picture}

Until this point, we have been analyzing field theories: we have started with a Drinfeld double, $\mathcal{D}(G)$, enumerated its irreps (the anyons of the model), and drawn connections between those irreps, through the usual labels (conjugacy classes and irreps of the group $G$). We have introduced the concept of "forbidding a label", which leads to forbidding anyons of the theory, and analyzed what other theories emerge, when the double is $\cD(S_3)$. 

In this section, we present a way to physically realize Drinfeld doubles: Kitaev's quantum double construction \cite{Kitaev03}. We introduce the physical lattice and the corresponding Hilbert space in Sec.~\ref{subsec:Hilbert_space}, the states of which correspond to the elements of a Drinfeld double. Based on this Hilbert space we provide a high-level picture of what we meant by "forbidding labels" in Secs.~\ref{sec:phase_transitions_general}-\ref{sec:phase_diagram}. Then, in Sec.~\ref{subsec:operators_Hamiltonian} we introduce a Hamiltonian to this system, the excitations of which will be the anyons, corresponding to the irreps of the algebra $\mathcal{D}(G)$. The introduction of this Hamiltonian will allow us to give a more concrete protocol to realize the phase transitions discussed in this paper.

\subsection{Projecting out part of the Hilbert space}
\label{subsec:Hilbert_space}

Quantum doubles, introduced by Kitaev \cite{Kitaev03}, are a way to realize the excitations of a Drinfeld double of group $G$, $\mathcal{D}(G)$ in a many-body, nearest-neighbor interacting lattice system. We assign a qudit Hilbert space of dimension $|G|$ to every edge of the lattice, as well as a direction to each edge, pointing from one end of the edge towards the other. 

The lattice geometry can be chosen to have an arbitrary graph structure, in this description we will focus on a square lattice. Then, anyons of this model live on sites formed by 6 qudits: the combination of the 4 qudits of a vertex and 4 qudits of a plaquette, with 2 qudits overlapping (see Fig.~\ref{fig:star_plaq_site}). The inner structure of a site, the fact that it is formed by a vertex and a plaquette, has significance, as the chargeon part of an anyon lives on the vertex, and the fluxon part lives on the plaquette.

\begin{figure}
\begin{centering}
\includegraphics[width=0.2\textwidth]{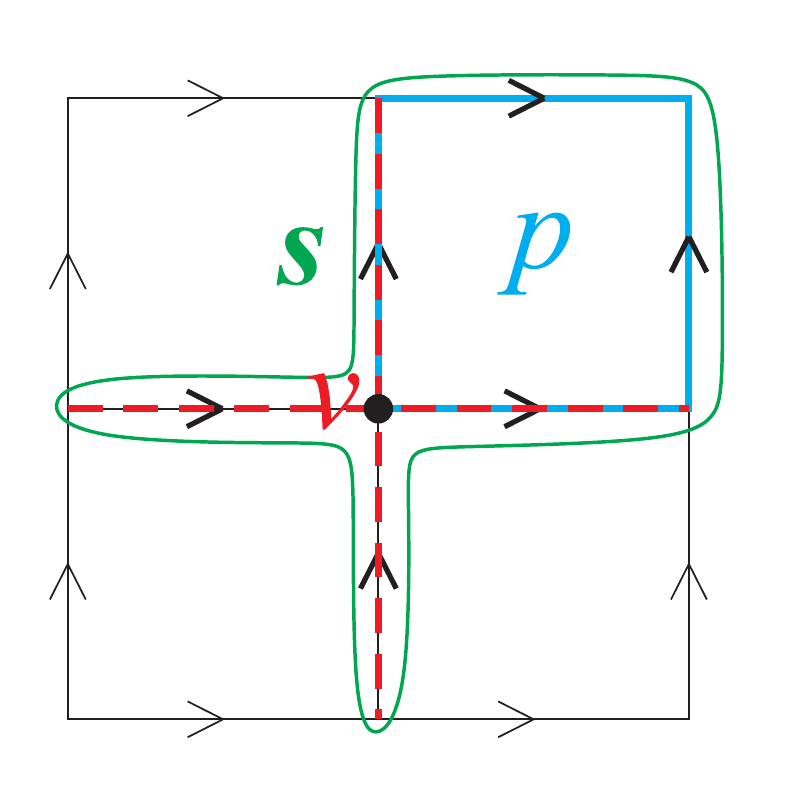}
\caption[A $6$-qudit site of a quantum double]{A site, formed by 6 qudits, is the combination of a vertex, $v$ (4 qudits) and a plaquette, $p$ (4 qudits), with 2 of the qudits overlapping.}
\label{fig:star_plaq_site}
\end{centering}
\end{figure}

Now, we're in a position to talk about the process of "forbidding", ubiquitously present in the current paper. In this lattice description, we can interpret the forbiddance of certain anyons or anyon labels by \emph{projecting out part of the Hilbert space}, the parts that contains the forbidden labels. As a result of such a complete projection, the theory will transition to a new theory, with a Hilbert space that is a subspace of the original one:
\bq \label{eq:phasetransition_projection}
&& \Pi_{\textrm{forbiddance}}: \\
&& \cH_{\cD(G)} = \textrm{span} \left\{ \bigotimes_{i \textrm{ qudits}} \ket{g_i}_i \right\}_{g_i \in G} \rightarrow  \cH_{\textrm{new}} \subset \cH_{\cD(G)}
\nonumber
\eq
where $\ket{g_i}_i$ is the state of the qudit of edge $i$.

For example, when the emergent theory is a quantum double itself:
\be
\cH_{\textrm{new}} = \textrm{span} \left\{ \bigotimes_{i \textrm{ qudits}} \ket{g_i}_i \right\}_{g_i \in H }
\ee
where $H \subset G$.

We can illustrate this point by considering the ground state of the Hamiltonian in a string-net representation \cite{LW05}. The ground state will be the sum of distinct string-net configurations, whose terms can be grouped into two separate sets: those containing strings of the forbidden anyon(s), and those that are formed only by the (still) allowed set of anyons. In this representation, the effect of the projection \eqref{eq:phasetransition_projection} on the ground state will simply correspond to removing one set of string-net terms from this sum.

Simple as it may sound, the operation \eqref{eq:phasetransition_projection} can't be realized unitarily on the system. Furthermore, it is unclear how to build such a projector from local projections unto lattice sites. While the local operations could successfully remove the forbidden excitations at all sites, they wouldn't remove states of the Hilbert space which have none of the forbidden excitations, but have condensation loops of the forbidden anyon (such as the string-net terms mentioned in the previous paragraph).

Instead, we will now introduce a way to approximate this projection. Here we will present this for $\cD (S_3)$, but all facts in Sec.~\ref{subsec:operators_Hamiltonian} can be stated for a quantum double of any general group, $\mathcal{D}(G)$ \cite{KL17}.

\subsection{Tuning a Hamiltonian}
\label{subsec:operators_Hamiltonian}

In the previous work~\cite{KL17}, we introduced two sets of orthogonal projectors unto the space of a single site of a quantum double: one set of charge projectors, and one set of flux projectors. The charge projectors act on vertices, the flux projectors act on plaquettes, while an anyon lives on the combination of the two (see Fig.~\ref{fig:star_plaq_site}). Hence, these projectors are all $4$-local.

Charge projectors correspond to irreps of the whole group ($S_3$ in this case), and the flux projectors to conjugacy classes of the group. Therefore, for $\cD(S_3)$, there are three charge projectors and three flux projectors. The two sets individually provide a full orthogonal resolution of the Hilbert space of a site. Furthermore, all those projectors commute pairwise.

Using these (or a subset of these) $4$-local projectors to construct a quantum double Hamiltonian will result in a peculiar property of the resulting model: anyons of the double won't be in one-to-one correspondence with the energy eigenspaces of this Hamiltonian \cite{KL17}. Depending on the exact \emph{flavor} of an anyon, the Hamiltonian would assign a different energy to it.

In order to obtain a fully topological model, we instead need to construct a Hamiltonian with $6$-local projectors, projecting unto all six qudits of a site. Then, we may utilize the $4$-local projectors by adding them to a quantum double with a topological ($6$-local) Hamiltonian. Tuning the couplings of the newly added projectors, we can achieve phase transitions as a limit of the original double $\cD(S_3)$.

\subsubsection{Sets of flux and charge projectors}
\label{subsubsec:4-body_projectors}

Let us first introduce $4$-local projectors unto the sites of $\cD(S_3)$. For $\mathcal{D}(S_3)$, the set of \emph{flux labels} on which the flux projectors are based are the conjugacy classes of $S_3$ (Eqs.~\eqref{eq:conj_classes_1}-\eqref{eq:conj_classes_3}). Then, we can introduce the projectors $$B_{C_e}, \; B_{C_x}, \; B_{C_y}$$ acting on plaquettes of the model. This set of projectors spans the plaquette (flux) space of a site, the flux labels $C_e$, $C_x$, and $C_y$ provide an orthogonal basis for all flux flavors. The rank of these projectors follow the rank of the conjugacy classes they are based on: $1$, $3$ and $2$, in order.

The images of these projectors are straightforward:
\bq \label{eq:image_B_Ce}
\Im (B_{C_e}) &=& A \oplus B \oplus C \\
\Im (B_{C_x}) &=& D \oplus E \\
\Im (B_{C_y}) &=& F \oplus G \oplus H \label{eq:image_B_Cy}
\eq
where the anyon labels $A$--$H$ denote the subspace of a site corresponding to having that anyon present at that site. (See Table~\ref{tab:D(S3)_anyons} for the list of anyons of $\cD(S_3)$.)

The set of \emph{charge labels} on which the charge projectors are based are the irreps of $S_3$ (Table~\ref{tab:irreps_of_S3}). Then, we can introduce the three projectors $$A_{\Gamma_1}, \; A_{\Gamma_{-1}}, \; A_{\Gamma_2}$$ acting on vertices of the model. They are of rank $1$, $1$ and $4$, in order, following the (square) dimensionality of the irreps they are based on. These three projectors span the vertex (charge) space of a site, and the charge labels $\Gamma_1^{S_3}$, $\Gamma_{-1}^{S_3}$ and $\Gamma_2^{S_3}$ provide an orthogonal basis for all charge flavors. This resolution of the Hilbert space of a site is complementary to the resolution provided by the three flux projectors. 

For non-Abelian doubles, the images of these projectors are less trivial than those of the flux projectors. For $\cD(S_3)$ they are:
\bq \label{eq:image_A_Gamma1}
\Im (A_{\Gamma_1}) &=& A \oplus D_1 \oplus F_1 \\
\Im (A_{\Gamma_{-1}}) &=& B \oplus E_1 \oplus F_2 \\
\Im (A_{\Gamma_2}) &=& C \oplus D_2 \oplus E_2 \oplus G \oplus H \label{eq:image_A_Gamma2}
\eq
where anyon labels with subscripts (e.g. $D_1$, $E_1$, $F_2$) denote orthogonal \emph{flavors} of anyons on a site \cite{KL17}.

Equations \eqref{eq:image_B_Ce}-\eqref{eq:image_B_Cy} and \eqref{eq:image_A_Gamma1}-\eqref{eq:image_A_Gamma2} are summarized in Fig.~\ref{fig:anyon-splitting-2}.

\begin{figure}
\begin{centering}
\includegraphics[width=0.4\textwidth]{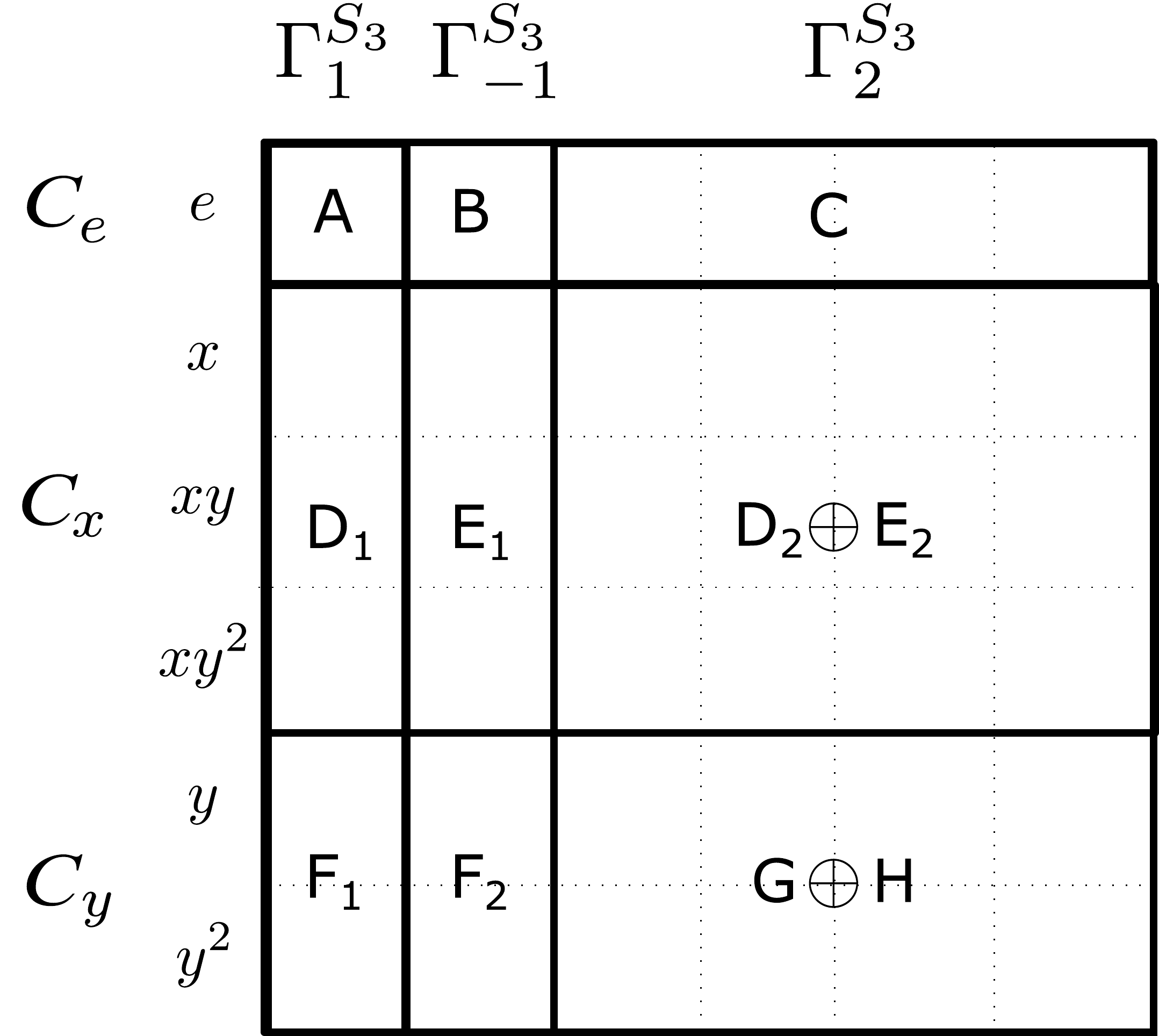}
\caption[Anyon splitting diagram, based on Fig.~\ref{fig:orth_resolution}, with labels for the different charge flavors of anyons]{Anyon splitting diagram, based on Fig.~\ref{fig:orth_resolution}. Here we added labels for the different charge flavors of anyons ($D_1$, $D_2$, $E_1$, $E_2$, and $F_1$, $F_2$) which are differentiated by the $4$-local charge projectors \eqref{eq:image_A_Gamma1}-\eqref{eq:image_A_Gamma2}.}
\label{fig:anyon-splitting-2}
\end{centering}
\end{figure}

It is key throughout these sections that (some of) the $4$-local projectors presented here are \emph{not topological}. Even though they form an orthogonal and commuting set, in their images certain anyons will split, depending on the flavor of the anyon. Physically, this means that adding such a $4$-local projector to the Hamiltonian of $\cD(S_3)$ will result in the new Hamiltonian distinguishing between flavors of the same anyon, i.e. yielding different energies based on the exact inner state of the anyon \cite{KL17}.

For example, anyon $D$ has two flavors: $D_1$ and $D_2$. While $D_1$ is in the image of $A_{\Gamma_1}$ (Eq.~\eqref{eq:image_A_Gamma1}), thus it is a chargeon, the other flavor $D_2$ is in the image of the $A_{\Gamma_2}$ projector (orthogonal to $A_{\Gamma_1}$) therefore it is should be considered a dyon (Eq.~\eqref{eq:image_A_Gamma2}). The standard Hamiltonian introduced in Ref.~\cite{Kitaev03}
\be \label{eq:4body_H}
H = - \sum_v A^{(v)}_{\Gamma_1} - \sum_p B^{(p)}_{C_e}
\ee
assigns energy $-1$ to the $D_1$ flavor, and energy $0$ to the $D_2$ flavor of anyon $D$.

In contrast, the purely topological Hamiltonian is $6$-local:
\be \label{eq:H_6body}
H = - \sum_{s=(v,p)} A^{(v)}_{\Gamma_1} B^{(p)}_{C_e}.
\ee
Notice, how the only difference from the \eqref{eq:4body_H} Hamiltonian is that here we multiply the $4$-local vertex and plaquette projectors to form a $6$-local projector unto sites. (This is simple for this case, when we only wish to project out the vacuum, and is less trivial for other, non-vacuum states, see Ref.~\cite{KL17}.)

The \eqref{eq:H_6body} Hamiltonian only distinguishes the single vacuum state by giving it $-1$ energy, while treating all states of excitations equally: it simply assigns energy $0$ to all (non-vacuum) anyons.

\subsubsection{Adding $4$-local terms to the Hamiltonian}

Let's take the fully topological Hamiltonian of $\cD(S_3)$, Eq.~\eqref{eq:H_6body}, which doesn't distinguish between flavors of the same anyon. Now, we may add the $4$-local terms introduced in in Sec.~\ref{subsubsec:4-body_projectors}, they all commute with both the vacuum projector and each other \cite{KL17}:
\bq \label{eq:Hamiltonian}
H &=& - \sum_{s=(v,p)} A^{(v)}_{\Gamma_1} B^{(p)}_{C_e} \\
&&+ \sum_v \left( \beta A^{(v)}_{\Gamma_{-1}} + \gamma A^{(v)}_{\Gamma_2} \right) + \sum_p \left( \epsilon B^{(p)}_{C_x} + \nu B^{(p)}_{C_y} \right) .
\nonumber
\eq
We could also add $4$-local terms for $A_{\Gamma_1}$ and $B_{C_e}$, but as shortly demonstrated, it is not necessary for our purposes.

The prefactors in front of each projector are tunable energy parameters (the notation follows that of Ref.~\cite{KL17}). As long as none of the parameters are decreased below $-0.5$ (more precisely, no combination of charge and flux projectors together are below $-1$), the ground space of this Hamiltonian coincides with the vacuum state of $\cD(S_3)$. Thus, all anyons of the Drinfeld double are excitations of the Hamiltonian \eqref{eq:Hamiltonian}.

\subsubsection{Energy-suppression of conjugacy classes or irreps} \label{subsubsec:tuning_parameters}

Engineering the Hamiltonian \eqref{eq:Hamiltonian}, and then tuning some of its parameters to be very large, will practically forbid anyons and anyon flavors related to the projector (or label of) we are tuning. In an environment with finite temperature, thermal processes that cost energy $E \gg k_B T$ will be suppressed. Thus, processes leading to the creation of certain anyons would cost too much energy, and those processes become thermally forbidden. This procedure is a way to physically realize the phase transitions discussed in Sec.~\ref{sec:phase_diagram}.

For example, tuning the parameter $\epsilon \rightarrow \infty$ in the Hamiltonian \eqref{eq:Hamiltonian} will result in the conjugacy class $C_x$ becoming forbidden in the theory, and lead to the analysis presented in Sec.~\ref{subsubsec:DZ3}. Another example is when we tune the parameter $\gamma \rightarrow \infty$, which will lead to the forbiddance of the irrep $\Gamma_2$, analyzed in Sec.~\ref{subsec:su2_4_2}. In short, tuning any combination of the parameters in \eqref{eq:Hamiltonian} will lead to the (joint) forbiddance of labels, as investigated in Sec.~\ref{sec:phase_diagram} and in Appendix~\ref{sec:all_theories}.

However, it is key to understand that in order to induce the phase transition, we need to \emph{completely project out} certain parts of the Hilbert space, as discussed in Sec.~\ref{subsec:Hilbert_space}. This will only happen when the tuned parameter reaches infinity. As an example, the Hamiltonian for which $\epsilon\gg1$ is not the same Hamiltonian as $\epsilon \to \infty$. Consider the \eqref{eq:Hamiltonian} Hamiltonian along the path $\{\beta=\gamma=\nu=0,\epsilon\}$ in which $\epsilon$ increases starting from 0. For $\epsilon=0$, the Hamiltonian is the fully topological \eqref{eq:H_6body} and by definition in the $\mathcal{D}(S_3)$ phase. As $\epsilon$ increases, the Hamiltonian remains in the $\mathcal{D}(S_3)$ phase since it is connected to the $\epsilon=0$ topological Hamiltonian by a smooth path \emph{without closing the gap} between the ground space and first-excited states. However, the \eqref{eq:Hamiltonian} Hamiltonian for which $\epsilon\to\infty$ is in the $\mathcal{D}(\mathbb{Z}_3)$ phase, as shown in Sec.~\ref{subsubsec:DZ3}. In other words, we expect the \eqref{eq:Hamiltonian} Hamiltonian $H(\epsilon\to\infty)$ to be renormalized to the topological Hamiltonian of $\mathcal{D}(\mathbb{Z}_3)$ by an RG flow. This discussion highlights the fact that it is possible to escape from a quantum phase by taking a limit, which makes a quantum phase an open set in the mathematical sense. Our picture of the phase transition is represented on Fig.~\ref{fig:Hamiltonian-path}. 

\begin{figure}
\begin{centering}
\includegraphics[width=0.8\columnwidth]{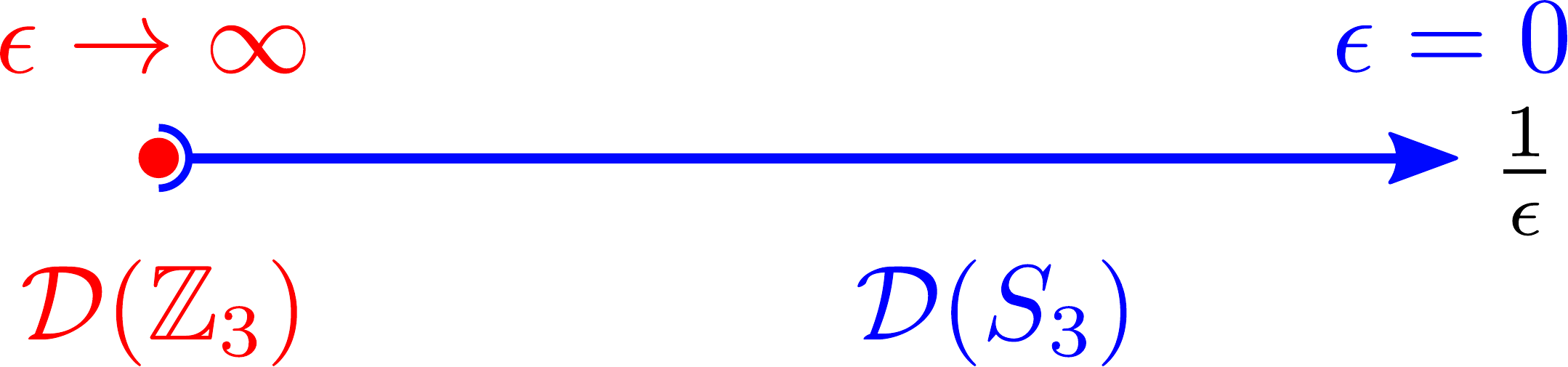}
\caption[The quantum phases along the path $\{\beta=\gamma=\nu=0,\epsilon\}$]{The quantum phases along the path $\{\beta=\gamma=\nu=0,\epsilon\}$. $\epsilon =0$ is in the topological Hamiltonian of the $\mathcal{D}(S_3)$ state. $\epsilon \to \infty$ is the topological Hamiltonian of the quantum phase $\mathcal{D}(\mathbb{Z}_3)$. Note that every point corresponds to a frustration-free local Hamiltonian with commuting terms. Aside from the point $\epsilon \to \infty$, all other Hamiltonians are in the $\mathcal{D}(S_3)$ phase.}
\label{fig:Hamiltonian-path}
\end{centering}
\end{figure}

\section{Discussion}
\label{sec:conclusions}

In this paper we introduced a framework and tools to describe quantum phase transitions in a family of topological field theories: quantum doubles. The physical processes we found can be described using terms borrowed from Landau's theory of phase transitions. We characterized phase transitions using anyonic data: fusion rules, quantum dimensions, and the $S$-matrix of a theory. However, this information together can not always distinguish all distinct theories, and a point of improvement could be to incorporate the $T$-matrix into our protocol.

A similar analysis can be done for any quantum double, $\cD(G)$. As quantum doubles of (Abelian) cyclic groups, $\cD(\mathbb{Z}_d)$ are the simple juxtaposition of $\mathbb{Z}_d$ charges and $\mathbb{Z}_d$ fluxes, with all anyons having quantum dimension $1$, the theories emerging there will always be quantum doubles of subgroups. Both in these cases and for the non-Abelian $\cD(S_3)$, the number of anyons and their quantum dimensions only decrease or don't change, but never increase (e.g. due to two anyons combining). It is an interesting question whether for more complex Abelian doubles (e.g. for $\cD(\mathbb{Z}_d \times \mathbb{Z}_k)$) this would still hold, or would there be processes leading to the (non-trivial) combination of Abelian anyons, resulting in the emergence of a non-Abelian theory \cite{WLW+08,FML+16}.

\subsection{The significance of $SU(2)_4$}

One of the theories that emerged from $\cD(S_3)$ is $SU(2)_4$, which is a non-trivial occurrence. It has recently been shown that anyons of $SU(2)_4$ can be utilized for universal quantum computation \cite{LBF+15} in that model. Our results show that it may be possible to realize $SU(2)_4$ on a lattice, by constructing $\cD(S_3)$ through a $6$-local Hamiltonian and introducing energy suppression through $4$-local terms.

While $6$-body terms are still too many, considering there exist theories with topological order with only $2$-body terms \cite{Kitaev06,BFB+11}, it is an improvement over similar proposals for universality. Levin-Wen models require $12$-body interactions \cite{LW05}, and the simplest double $\cD(A_5)$ that is proven to have universal braiding, while it requires $6$-body interactions, needs physical qudits with $60$ distinct levels \cite{OP99,Preskill98_paper}. In contrast, $\cD(S_3)$ and the emergent $SU(2)_4$ only requires $6$-level qudits for their lattice realization.

\subsection{Emerging chirality}

The emergence of chirality through $SU(2)_4$, while the original theory $\cD(S_3)$ was non-chiral, is interesting. Notice, how in both cases yielding $SU(2)_4$ (forbidding conjugacy class $C_y$, or the irrep $\Gamma_2$), a process of (spontaneous) symmetry breaking happens, before the model finally transitions into the chiral $SU(2)_4$.

It seems this partial forbiddance, together with the unique properties of anyons $D$ and $E$ in the original $\cD(S_3)$ theory, are what cause the emergence of chirality. Anyons $D$, $E$ are peculiar in how their charge flavors split: recall that while the normalizer of $x$ is isomorphic to the normalizers of $xy$ and of $xy^2$, they are not identical. While the irreps of each of those normalizers will split between the irreps of $S_3$ (e.g. between $\Gamma_1^{S_3}$ and $\Gamma_2^{S_3}$ for anyon $D$), the manner how they split is different for each normalizer, thus for each flux label considered. This causes anyons $D$ and $E$ to react non-trivially to the forbiddances discussed in this work.

An illustration of a process that breaks the symmetry, after the conjugacy class $C_y$ disappears, is the following counterclockwise and clockwise braiding of flux labels $x$ and $y$:
\bq
R \ket{x,y} &=& \ket{xyx^{-1},x} = \ket{y^2,x} , \\
R^{-1} \ket{x,y} &=& \ket{y,y^{-1}xy} = \ket{y,xy^2} ,
\eq 
where both of these exchanges only switch flux flavors, thus do nothing on the anyon-level as long as all the elements in the $C_x$ and $C_y$ conjugacy class (all the flavors) are present in the theory. However, as soon as we forbid the conjugacy class $C_y$, we break the symmetry of the $C_x$ conjugacy class: two of the flavors (say $xy$ and $xy^2$) become forbidden, while only one flavor continues to live in the theory ($x$). Thus, in the exchange scenario above, the second (clockwise) exchange would lead to an unphysical state, as $xy^2$ is forbidden, while the first (counterclockwise) exchange doesn't change the flavor $x$, and it can easily go through. Thus, we broke the exchange symmetry in this model.

This argument doesn't provide a full explanation of how chirality emerges, especially as one of the flux labels used ($y$) in the argument is forbidden in the new theory. It only provides one example of breaking the symmetry of a fundamental process. We leave it to future work to investigate how exactly chirality emerges as a result of the properties discussed here.

\subsection{Difference between Projection and Energy Suppression}

In Sec.~\ref{sec:lattice_picture} we discussed how phase transitions could be realized through projecting out part of a Hilbert space vs. using a modified Hamiltonian to suppress anyons through their energy couplings. One difference between these two methods concerns the ground space of the new theory.

The projection operator will clearly change the ground space, therefore the new theory will have the correct ground state. Meanwhile, tuning the parameters of a commuting Hamiltonian, as done during the energy suppression, won't change the ground space of the model. This is especially interesting in light of Ref.~\cite{WN90}, which establishes a connection between anyons of a model and the ground-space degeneracy, thus it hints at the fact that a different set of anyons should result in a changed ground space.

We refer to our analysis of the process of energetic suppression in Sec.~\ref{subsubsec:tuning_parameters}, where we argued that energy suppression will move a theory towards a phase transition, but won't necessarily take it over the critical point. Thus, the fact that the ground space throughout this procedure remains unchanged, is in agreement with our understanding of the distinction between energetic suppression vs. projection of the Hilbert space.

Another consequence of the subtle difference between projection and energy suppression is the following. In certain sections the reader might expect our analysis to yield a different theory than the one presented, for example, in Sec.~\ref{subsec:Z3_C}. There we forbid all non-trivial flux labels, and as a result we would expect that the pure chargeon theory of $\cD(S_3)$ will arise: the $\Phi$--$\Lambda$ theory \cite{WBI+14}. Following our analysis, instead, leads us to the emergence of $\mathbb{Z}_3$. Why is this?

It turns out that in some of the cases investigated in this paper, performing the projection on the Hilbert space vs. energy suppression will result in completely different theories. In the above example, energetically suppressing all anyons with non-trivial fluxes, through the Hamiltonian, will result in the $\Phi$--$\Lambda$ theory, living in the complete Hilbert space of $\cD(S_3)$. However, when we \emph{project out} parts of the Hilbert space related to the forbidden anyons, we will transition into a theory that needs to live in this new subspace of a Hilbert space. The $\Phi$--$\Lambda$ theory needs support on the full Hilbert space, and it can't manifest in this subspace. Therefore, the theory after the projection will transition into $\mathbb{Z}_3$ instead, for which the new Hilbert space is a sufficient support. This is reminiscent of two-dimensional topological theories that can only live on boundaries of three-dimensional theories.

\section{Acknowledgments}

We thank John Preskill, Alexei Kitaev, David Aasen, Sujeet Shukla, and Burak Sahinoglu for helpful discussions. We acknowledge funding provided by the Institute for Quantum Information and Matter, an NSF Physics Frontiers Center (NSF Grant PHY-1125565) with support of the Gordon and Betty Moore Foundation (GBMF-2644). OLC is partially supported by the Natural Sciences and Engineering Research Council of Canada (NSERC).

\appendix

\section{Additional details of quantum doubles}
\label{sec:fusion_Smatrix_DS3}

\subsection{$S$-matrix of $\mathcal{D}(\mathbb{Z}_{3})$}

The $S$-matrix of the theory is \cite{Bonderson07}:
\be \label{eq:DZ3_Smatrix}
S= \frac{1}{3}\left[\begin{array}{ccccccccc}
1 & 1 & 1 & 1 & 1 & 1 & 1 & 1 & 1 \\
1 & 1 & 1 & \omega & \bar{\omega} & \omega & \omega & \bar{\omega} & \bar{\omega} \\
1 & 1 & 1 & \bar{\omega} & \omega & \bar{\omega} & \bar{\omega} & \omega & \omega \\
1 & \bar{\omega} & \omega & 1 & 1 & \bar{\omega} & \omega & \bar{\omega} & \omega \\
1 & \omega & \bar{\omega} & 1 & 1 & \omega & \bar{\omega} & \omega & \bar{\omega} \\
1 & \bar{\omega} & \omega & \omega & \bar{\omega} & 1 & \bar{\omega} & \omega & 1 \\
1 & \bar{\omega} & \omega & \bar{\omega} & \omega & \omega & 1 & 1 & \bar{\omega} \\
1 & \omega & \bar{\omega} & \omega & \bar{\omega} & \bar{\omega} & 1 & 1 & \omega \\ 
1 & \omega & \bar{\omega} & \bar{\omega} & \omega & 1 & \omega & \bar{\omega} & 1
\end{array} \right] ,
\ee
where rows and columns correspond to the elements in the following order:
$$\{ 1,e_1,e_2,m_1,m_2,e_1 m_1, e_2 m_1, e_1 m_2, e_2 m_2 \}$$

\subsection{Fusion rules and $S$-matrix of $\mathcal{D}(S_{3})$}

The fusion rules of the 8 anyons of $\mathcal{D}(S_3)$ are shown in Table~\ref{tab:fusion_D(S3)}.
\begin{table*}[t]
	\centering
		\begin{tabular}{ |c||c|c|c|c|c|c|c|c| } 
			\hline
			 & A & B & C & D & E & F & G & H \\
			\hline
			\hline
			A & A & B & C & D & E & F & G & H \\ 
			\hline
			B & B & A & C & E & D & F & G & H \\ 
			\hline
			C & C & C & A $\oplus$ B $\oplus$ C & D $\oplus$ E & D $\oplus$ E & G $\oplus$ H & F $\oplus$ H & F $\oplus$ G \\
			\hline
			D & D & E & D $\oplus$ E & A $\oplus$ C $\oplus$ F $\oplus$ G $\oplus$ H & B $\oplus$ C $\oplus$ F $\oplus$ G $\oplus$ H & D $\oplus$ E & D $\oplus$ E & D $\oplus$ E \\
			\hline
			E & E & D & D $\oplus$ E & B $\oplus$ C $\oplus$ F $\oplus$ G $\oplus$ H & A $\oplus$ C $\oplus$ F $\oplus$ G $\oplus$ H & D $\oplus$ E & D $\oplus$ E & D $\oplus$ E \\
			\hline
			F & F & F & G $\oplus$ H & D $\oplus$ E & D $\oplus$ E & A $\oplus$ B $\oplus$ F & H $\oplus$ C & G $\oplus$ C \\
			\hline
			G & G & G & F $\oplus$ H & D $\oplus$ E & D $\oplus$ E & H $\oplus$ C & A $\oplus$ B $\oplus$ G & F $\oplus$ C \\
			\hline
			H & H & H & F $\oplus$ G & D $\oplus$ E & D $\oplus$ E & G $\oplus$ C & F $\oplus$ C & A $\oplus$ B $\oplus$ H \\
			\hline
		\end{tabular}
		\caption{Fusion rules of anyons in $\mathcal{D}(S_3).$}
		\label{tab:fusion_D(S3)}
\end{table*}

The $S$-matrix of the theory \cite{BSW11}:
\be \label{eq:D(S3)_Smatrix}
S= \frac{1}{6}\left[\begin{array}{cccccccc}
1 & 1 & 2 & 3 & 3 & 2 & 2 & 2 \\
1 & 1 & 2 & -3 & -3 & 2 & 2 & 2 \\
2 & 2 & 4 & 0 & 0 & -2 & -2 & -2 \\
3 & -3 & 0 & 3 & -3 & 0 & 0 & 0 \\
3 & -3 & 0 & -3 & 3 & 0 & 0 & 0 \\
2 & 2 & -2 & 0 & 0 & 4 & -2 & -2 \\
2 & 2 & -2 & 0 & 0 & -2 & 4 & -2 \\
2 & 2 & -2 & 0 & 0 & -2 & -2 & 4 
\end{array} \right] .
\ee

\section{Phases based on subgroups of $S_{3}$}
\label{sec:all_theories}

There are several ways we can forbid a flux or charge label, or a set of labels in the theory $\mathcal{D}(S_3)$. In the main text we presented three of these cases, here we list and present all possible emerging theories, with full proofs and derivations. Most of the emergent theories are quantum doubles of a subgroup of $S_3$, like $\mathcal{D}(\mathbb{Z}_3)$ in Sec.~\ref{subsubsec:DZ3}, or the charge/flux sector of one. The only exception is $SU(2)_4$, presented in Secs.~\ref{subsec:su2_4_1}-\ref{subsec:su2_4_2} of the main text.

\subsection{Forbidding $C_y$ and $\Gamma_2$ jointly leads to $\mathcal{D}(\mathbb{Z}_{2})$}
\label{subsec:D(Z2)}

Forbidding both the conjugacy class $C_y$ and the 2-dimensional irrep $\Gamma_2$ will lead to the theory $\mathcal{D}(\mathbb{Z}_2)$. For an overview of this theory, see Sec.~\ref{subsubsec:toric_code}.

The process $\mathcal{D}(S_3)$ undergoes to form $\mathcal{D}(\mathbb{Z}_2)$ is the following.

\textit{Step 1 ---} Anyons $C,F,G,H$ become forbidden and $D,E$ become partially forbidden (their charge flavors related to $\Gamma_2$ are removed).

\textit{Step 2 ---} A (further) partial forbiddance of $D$ and $E$ is induced. This time it is a spontaneous symmetry breaking of their flux flavors: only one of the flavors $x$, $xy$ or $xy^2$ will remain. (This is the same process when flux flavors of $D,E$ became forbidden after forbidding the $C_y$ conjugacy class only, in Sec.~\ref{subsec:su2_4_1}. For more details, please refer to that section.)

At the end, the remaining set of anyons is: $\{ A,B,D,E \}$, where $D$ and $E$ has decreased dimensions $1$, due to the series of partial forbiddances these anyons underwent. These anyons correspond to the anyons of $\mathcal{D}(\mathbb{Z}_2)$ with
\bq
\nonumber
A &=& 1 \\
\nonumber
B &=& \mathbf{e} , \\
\nonumber
D &=& \mathbf{m} , \\
\nonumber
E &=& \mathbf{em} .
\eq
For an overview of these physical processes, see Fig.~\ref{fig:Cy_Gamma2_limit}.

\begin{figure}
\begin{centering}
\includegraphics[width=0.5\textwidth]{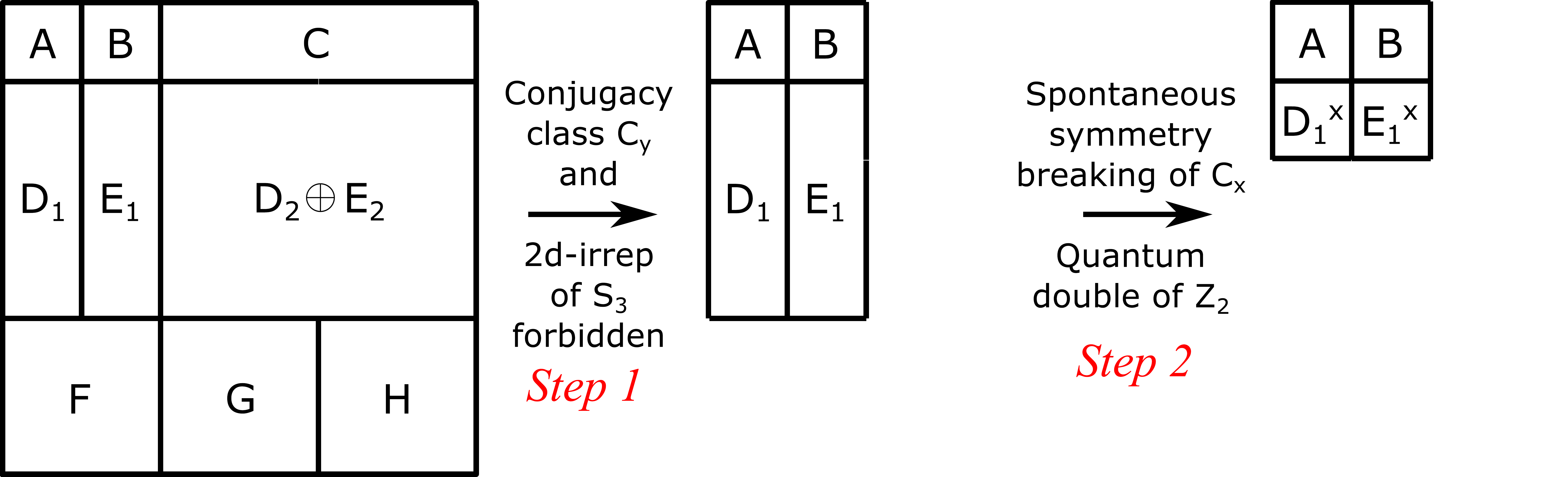}
\caption[Pictorial representation of the process transforming $\mathcal{D}(S_3)$ to $\mathcal{D}(\mathbb{Z}_2)$]{Pictorial representation of the process transforming $\mathcal{D}(S_3)$ to $\mathcal{D}(\mathbb{Z}_2)$. The superscripts $x$ denote the components of anyons $D$ and $E$ after the symmetry breaking.}
\label{fig:Cy_Gamma2_limit}
\end{centering}
\end{figure}

\subsubsection*{Proof}

We can prove this correspondence by following the protocol detailed in Sec.~\ref{subsec:protocol}. We need to start by truncating the fusion rules to the remaining set of anyons $\{ A,B,D,E \}$, then using these to reconstruct the new $S$-matrix. This will yield the $S$-matrix of $\mathcal{D}(\mathbb{Z}_2)$:
\be
S=\frac{1}{2}\left[\begin{array}{cccc}
1 & 1 & 1 & 1 \\
1 & 1 & -1 & -1 \\
1 & -1 & 1 & -1 \\
1 & -1 & -1 & 1
\end{array} \right] .
\ee

Arriving at this matrix (rather than, e.g. the $S$-matrix of the double semion model) is non-trivial, as such symmetries could be lost during our protocol. We'd like to remind the reader that when equivalent $S$-matrices emerge (as in, they produce the same fusion rules), we make our choice based on the symmetries of the $S$-matrix of the original $\mathcal{D}(S_3)$, which is what we did here. \qed

\subsection{Forbidding $C_x$, $C_y$, $\Gamma_2$ jointly leads to $\mathbb{Z}_2$}

It is trivial to see this transition. Anyons $C,D,E,F,G,H$ will become forbidden, leaving the set $\{ A,B \}$. Clearly, $B$ will act as a single $\mathbb{Z}_2$ charge in this model. For a pictorial argument, see Fig.~\ref{fig:Cy_Gamma2_Cx_limit}.

\begin{figure}
\begin{centering}
\includegraphics[width=0.5\textwidth]{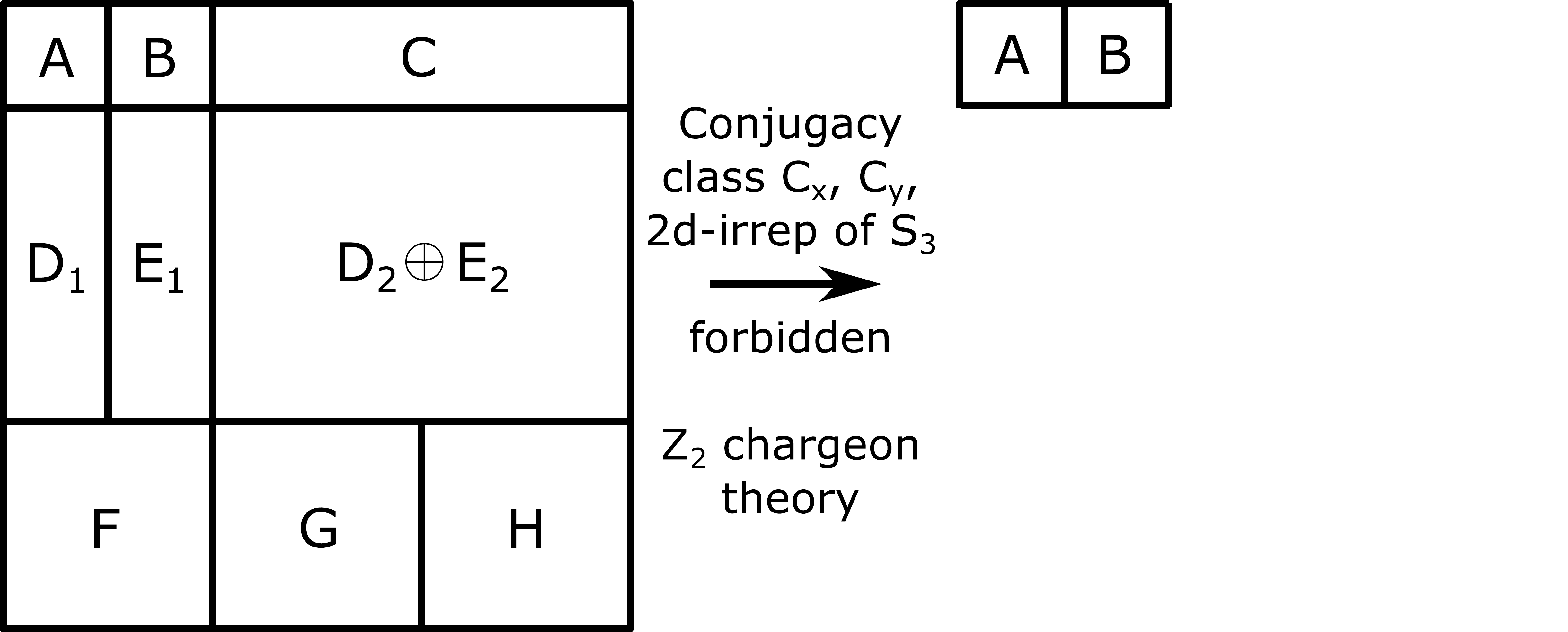}
\caption{Forbidding conjugacy classes $C_x$ and $C_y$, and irrep $\Gamma_2$ of $S_3$ transforms $\mathcal{D}(S_3)$ to $\mathbb{Z}_2$.}
\label{fig:Cy_Gamma2_Cx_limit}
\end{centering}
\end{figure}

\subsection{Forbidding $C_y$, $\Gamma_{-1}$, $\Gamma_2$ jointly leads to $\mathbb{Z}_2$}

\textit{Step 1 ---} This leads to the forbiddance of anyons $B,C,E,F,G,H$ and the partial forbiddance of $D$ (the charge flavor related to $\Gamma_2$ becomes forbidden).

\textit{Step 2 ---} The forbiddance of conjugacy class $C_y$ induces the partial forbiddance of anyon $D$ (this is a spontaneous symmetry breaking, for more details refer to the similar process in Sec.~\ref{subsec:su2_4_1}).

The set $\{ A,D \}$ remains, where the quantum dimension of $D$ has been decreased to $1$ due to the series of partial forbiddances the anyon underwent. For an overview of this process, see Fig.~\ref{fig:Cy_Gamma2_Gamma-1_limit}.

\begin{figure}
\begin{centering}
\includegraphics[width=0.5\textwidth]{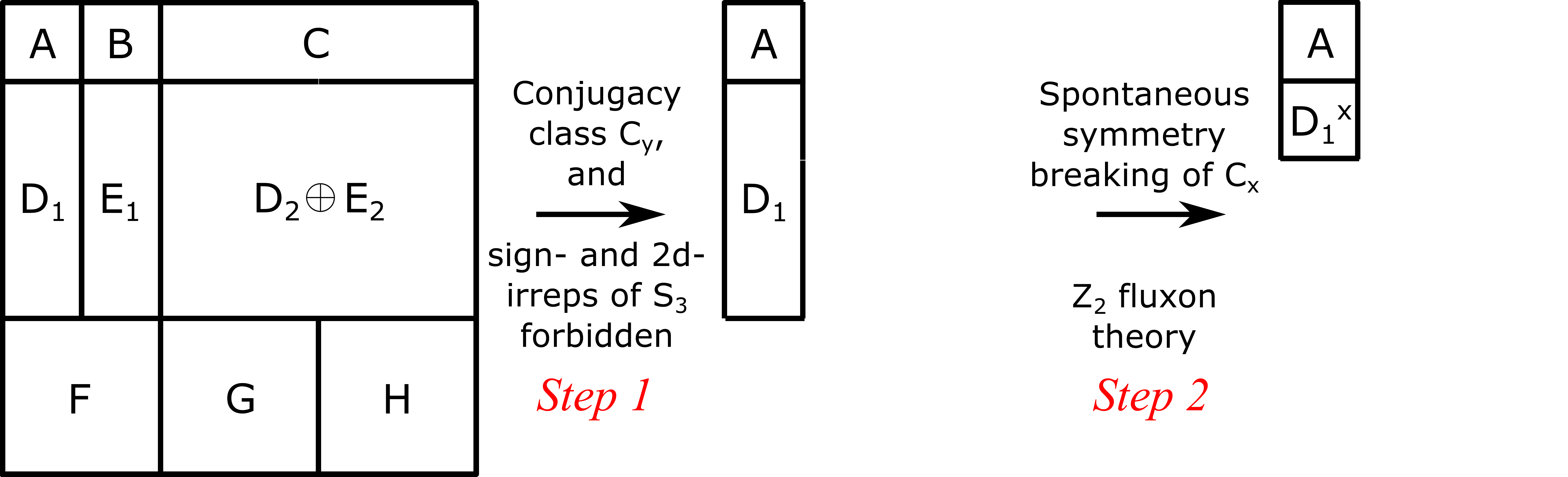}
\caption[Forbidding conjugacy class $C_y$, and irreps $\Gamma_{-1}$ and $\Gamma_2$ of $S_3$ transforms $\mathcal{D}(S_3)$ to $\mathbb{Z}_2$]{Forbidding conjugacy class $C_y$, and irreps $\Gamma_{-1}$ and $\Gamma_2$ of $S_3$ transforms $\mathcal{D}(S_3)$ to $\mathbb{Z}_2$. The superscript $x$ denotes the component of anyon $D$ after the symmetry breaking.}
\label{fig:Cy_Gamma2_Gamma-1_limit}
\end{centering}
\end{figure}

\subsubsection*{Proof}

This process can be proven by tying it to the proof of the emergent $\mathcal{D}(\mathbb{Z}_2)$ theory in Sec.~\ref{subsec:D(Z2)}. It's clear that if forbidding the labels $C_y$ and $\Gamma_2$ together led to the set of anyons $\{ A,B,D,E \}$ forming $\mathcal{D}(\mathbb{Z}_2)$, with the correct dimensions of $1$, then additionally forbidding the irrep $\Gamma_{-1}$ forbids the nontrivial charges in the model. This leaves the fluxon sector of $\mathcal{D}(\mathbb{Z}_2)$: $\{ A,D \}$, with the correct fusion rules and quantum dimensions for $\mathbb{Z}_2$.

Alternatively, it is also possible to prove the emergence of $\mathbb{Z}_2$ in this limit by dutifully following the procedure outlined in Sec.~\ref{subsec:protocol}. \qed

\subsection{Forbidding $\Gamma_{-1}$ and $\Gamma_2$ jointly leads to $\mathbb{Z}_2$}
\label{subsec:gamma-1_gamma2}

\textit{Step 1 ---} Forbidding $\Gamma_{-1}$ and $\Gamma_2$ leads to forbiddance of anyons $B,C,E,G,H$ and to the partial forbiddance of $D$ and $F$ (only the charge flavor related to $\Gamma_1$ remains), leaving the set of anyons $\{ A,D,F \}$.

\textit{Step 2 ---} This initial forbiddance will induce the condensation of $F$.

\textit{Step 3 ---} The set of remaining anyons are thus: $\{ A',D \}$, which form the $\mathbb{Z}_2$ theory (with multiplicity). 

For a pictorial argument showing this process, please refer to Fig.~\ref{fig:Gamma-1_Gamma2_limit}.

\subsubsection*{Parts of the Proof}

One can start to prove this by following the protocol of Sec.~\ref{subsec:protocol}.

\textit{Step 1 ---} We truncate the fusion rules for the set of remaining anyons, see the new rules in Table~\ref{tab:fusion_ADF}.

\begin{table}
	\centering
		\begin{tabular}{ |c||c|c|c| } 
			\hline
			 & A & D & F \\
			\hline
			\hline
			A & A & D & F \\ 
			\hline
			D & D & A $\oplus$ F & D \\  
			\hline
			F & F & D & A $\oplus$ F \\  
			\hline
		\end{tabular}
		\caption{Fusion rules of remaining particles of $\mathcal{D}(S_3)$, after forbidding irreps $\Gamma_{-1}$ and $\Gamma_2$.}
		\label{tab:fusion_ADF}
\end{table}

Then, we use the new set of fusion rules to reverse-engineer the new $S$-matrix for this theory, and get:
\be \label{eq:Smatrix_ADF}
S \propto \left[\begin{array}{ccc}
1 & \sqrt{3} & \sqrt{2} \\
\sqrt{3} & -3 & \sqrt{6} \\
\sqrt{2} & \sqrt{6} & 2  
\end{array} \right] .
\ee

\begin{lem}[Condensation (\textit{Step 2}, \textit{Step 3})] \label{lem:condensation_ADF}
Anyon $F$ condenses to the vacuum. The resulting theory is $\mathbb{Z}_2$, with multiplicity.
\end{lem}

\begin{proof}
From the $S$-matrix \eqref{eq:Smatrix_ADF} it's clear that the entries for anyons $A$ and $F$ are not independent of each other, indicating anyon $F$ will condense. This leads to the new theory having two distinct subspaces only: the space of $A$ and $F$ after condensation, having joint quantum dimension $\sqrt{3}$, and the space of anyon $D$ with dimension $\sqrt{3}$. This will be followed by a spontaneous symmetry breaking of $D$, motivated by how the loss of anyon $F$ always led to this process in previous cases (see Secs.~\ref{subsec:su2_4_1} and \ref{subsec:D(Z2)}), and we will get a theory with anyons $\{ A', D \}$, both having quantum dimension $1$, clearly the semion model.

It is possible to mathematically follow the condensation process of anyon $F$. However, due to it being a non-Abelian anyon with quantum dimension higher-than-$1$, to our knowledge it is not possible to model its condensation and keep track of the quantum dimensions of the resulting subspaces faithfully \emph{at the same time}. One can prove that $F$ will condense, by making the appropriate basis transformation on the $S$-matrix, but the information that the quantum dimensions of the resulting subspaces will match, will be lost.

However, even without the detailed mathematical argument for this condensation, the physical arguments are compelling enough to conclude that the emerging theory will be $\mathbb{Z}_2$, based on the set of anyons $\{ A',D \}$, formed after the condensation of $F$.
\end{proof}

\subsection{Forbidding $\Gamma_{-1}$ and $C_y$ jointly leads to $\mathbb{Z}_2$}

This process is very similar to the one detailed in Sec.~\ref{subsec:gamma-1_gamma2}, due to the exchange symmetry of anyons $C$ and $F$ in $\mathcal{D}(S_3)$.

\textit{Step 1 ---} Forbidding $\Gamma_{-1}$ and $C_y$ together will forbid anyons $\{ B,F,G,H \}$, and partially forbid anyon $E$.

\textit{Step 2 ---} Forbidding the irrep $\Gamma_{-1}$ induces anyon $C$ to lose its antisymmetric flavor, as well as it leads to the complete forbiddance of $E$ (all of its charge flavors are antisymmetric). This leaves the set of anyons $\{ A,C,D \}$.

Notice that due to the $C \leftrightarrow F$ interchangeability, this theory will undergo the same transitions as the one detailed in Sec.~\ref{subsec:gamma-1_gamma2}.

\textit{Step 3 ---} First, anyon $D$ will undergo a spontaneous symmetry breaking (induced by the loss of conjugacy class $C_y$).

\textit{Step 4 ---} Then anyon $C$ will condense to the vacuum.

\textit{Step 5 ---} After this, $\{ A',D \}$ will form $\mathbb{Z}_2$ (with multiplicity).

For an overview of this process, see Fig.~\ref{fig:Gamma-1_Cy_limit}.

\begin{figure}
\begin{centering}
\includegraphics[width=0.5\textwidth]{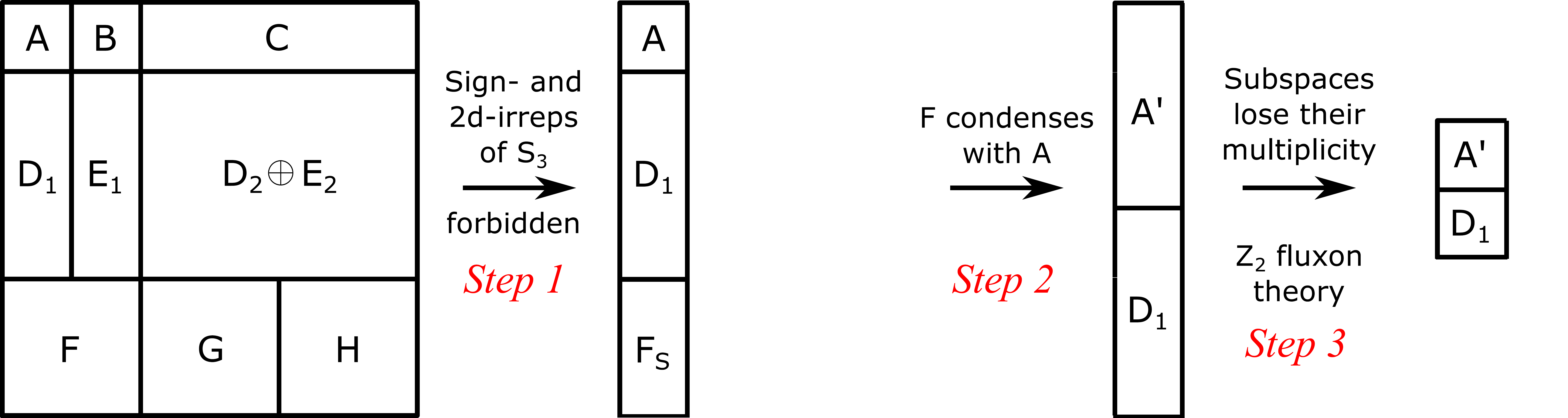}
\caption[Forbidding irreps $\Gamma_{-1}$ and $\Gamma_2$ of $S_3$ transforms $\mathcal{D}(S_3)$ to $\mathbb{Z}_2$]{Forbidding irreps $\Gamma_{-1}$ and $\Gamma_2$ of $S_3$ transforms $\mathcal{D}(S_3)$ to $\mathbb{Z}_2$. The subspace $F_S \equiv F_1$ of Fig.~\ref{fig:anyon-splitting-2}, $A'$ is the vacuum after the condensation of $F$.}
\label{fig:Gamma-1_Gamma2_limit}
\end{centering}
\end{figure}

\begin{figure*}
\begin{centering}
\includegraphics[width=\textwidth]{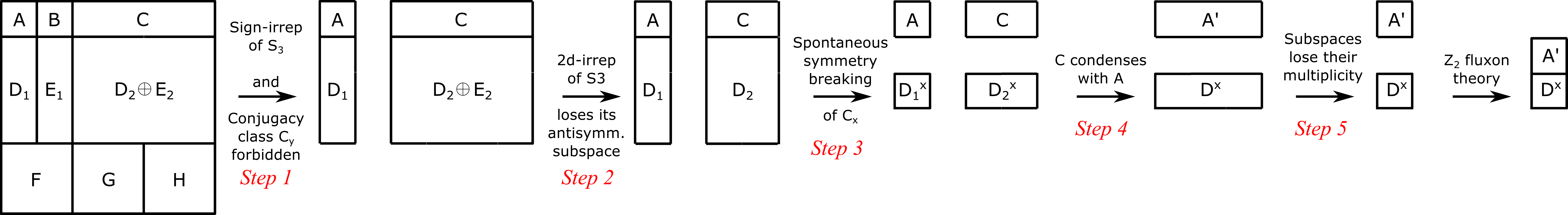}
\caption{Forbidding irrep $\Gamma_{-1}$ and conjugacy class $C_y$ of $S_3$ transforms $\mathcal{D}(S_3)$ to $\mathbb{Z}_2$. The superscript $x$ denotes the component of anyon $D$ after the symmetry breaking. $A'$ is the vacuum after the condensation of $C$.}
\label{fig:Gamma-1_Cy_limit}
\end{centering}
\end{figure*}

\subsubsection*{Proof}

Due to the $C \leftrightarrow F$ equivalence in the original $\cD(S_3)$ theory the proof is identical to the proof presented in Sec.~\ref{subsec:gamma-1_gamma2}. We need to start by truncating the original fusion rules to the new set $\{ A,C,D \}$, then follow the steps outlined there. \qed

\subsection{Forbidding $\Gamma_{-1}$ leads to $\mathbb{Z}_2$}

The double $\mathcal{D}(S_3)$ also transitions into $\mathbb{Z}_2$ when we simply forbid the irrep $\Gamma_{-1}$ alone.

\textit{Step 1 ---} $B$ becomes forbidden, and $E,F$ are partially forbidden.

\textit{Step 2 ---} The forbiddance of the irrep $\Gamma_{-1}$ additionally induces $C,G,H$ to lose their antisymmetric charge flavors, and the complete forbiddance of $E$ (as it only has antisymmetric charge flavors).

The remaining set of anyons is then: $\{ A,C,D,F,G,H \}$.

\textit{Step 3 ---} We will see that at this point anyons $C,F,G,H$ all condense to the vacuum.

\textit{Step 4 ---} We have two subspaces: $\{ A', D \}$, which will realize the $\mathbb{Z}_2$ theory.

For an overview of these processes, see Fig.~\ref{fig:Gamma-1_limit}.

\begin{figure*}
\begin{centering}
\includegraphics[width=\textwidth]{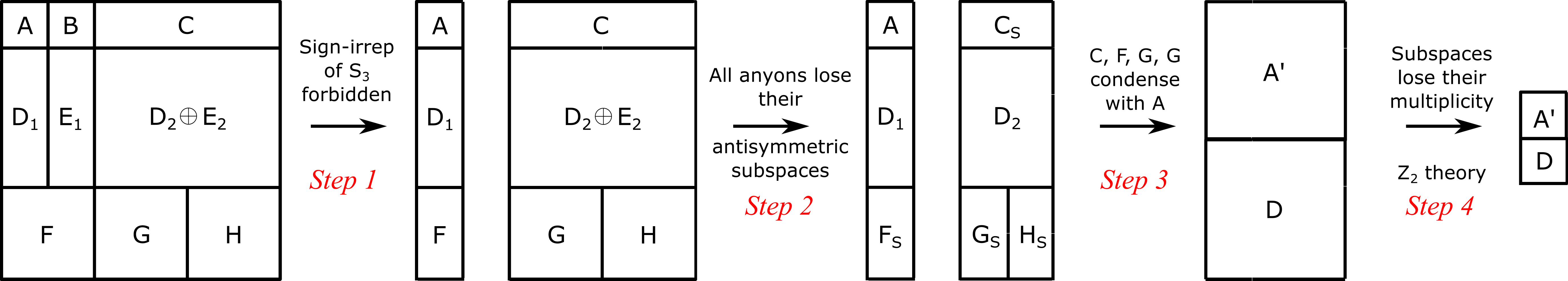}
\caption{Forbidding irrep $\Gamma_{-1}$ of $S_3$ transforms $\mathcal{D}(S_3)$ to $\mathbb{Z}_2$. The subscripts $S$ denote the symmetric subspaces of anyons ($F_S \equiv F_1$ of Fig.~\ref{fig:anyon-splitting-2}). $A'$ is the vacuum after the condensation of $C, F, G, H$.}
\label{fig:Gamma-1_limit}
\end{centering}
\end{figure*}

\subsubsection*{Parts of the Proof}

After the initial forbiddance processes, we are left with the set of anyons: $\{ A,C,D,F,G,H \}$.

\textit{Step 1 ---} Following our protocol, we truncate the fusion rules to this set, and construct the $S$-matrix. The truncated fusion rules are shown in Table~\ref{tab:fusion_ACDFGH}. 
\begin{table}
	\centering
		\begin{tabular}{ |c||c|c|c|c|c|c| } 
			\hline
			 & A & C & D & F & G & H \\
			\hline
			\hline
			A & A & C & D & F & G & H \\ 
			\hline
			C & C & A $\oplus$ C & D & G $\oplus$ H & F $\oplus$ H & F $\oplus$ G \\ 
			\hline
			D & D & D & A $\oplus$ C $\oplus$ F $\oplus$ G $\oplus$ H & D & D & D \\ 
			\hline
			F & F & G $\oplus$ H & D & A $\oplus$ F & C $\oplus$ H & C $\oplus$ G \\ 
			\hline
			G & G & F $\oplus$ H & D & C $\oplus$ H & A $\oplus$ G & C $\oplus$ F \\ 
			\hline
			H & H & F $\oplus$ G & D & C $\oplus$ G & C $\oplus$ F & A $\oplus$ H \\ 
			\hline
		\end{tabular}
		\caption{Fusion rules of remaining particles of $\mathcal{D}(S_3)$, after forbidding irrep $\Gamma_{-1}$.}
		\label{tab:fusion_ACDFGH}
\end{table}

\textit{Step 2 ---} The resulting $S$-matrix already shows us that anyons $C,F,G,H$ lost their antisymmetric subspaces:
\be
S \propto \left[\begin{array}{cccccc}
1 & \sqrt{2} & 3 & \sqrt{2} & \sqrt{2} & \sqrt{2} \\
\sqrt{2} & 2 & 3 \sqrt{2} & 2 & 2 & 2 \\
3 & 3 \sqrt{2} & -9 & 3 \sqrt{2} & 3 \sqrt{2} & 3 \sqrt{2} \\
\sqrt{2} & 2 & 3 \sqrt{2} & 2 & 2 & 2 \\
\sqrt{2} & 2 & 3 \sqrt{2} & 2 & 2 & 2 \\
\sqrt{2} & 2 & 3 \sqrt{2} & 2 & 2 & 2
\end{array} \right] .
\ee

\begin{lem}[Condensation (\textit{Step 3}, \textit{Step 4})]
Anyons $C,F,G,H$ all condense to the vacuum. The resulting theory is $\mathbb{Z}_2$, with multiplicity.
\end{lem}

\begin{proof}
Notice, that the entries corresponding to anyons $A,C,F,G,H$ are identical, up to a constant, while the only linearly independent entries are those of anyon $D$. From this, we conclude that excitations $C,F,G,H$ will all condense to the vacuum, forming a $3$-dimensional subspace. Then, we are left with two $3$-dimensional subspaces: one for the vacuum after condensation, $A'$, and one for anyon $D$. These two will form a $\mathbb{Z}_2$ theory. (For a more detailed argument, see the proof of Lemma~\ref{lem:condensation_ADF}.)
\end{proof}

\subsection{Forbidding $C_x$ and $C_y$ jointly leads to $\mathbb{Z}_3$} \label{subsec:Z3_C}

When we forbid conjugacy classes $C_x$ and $C_y$:

\textit{Step 1 ---} Anyons $D,E,F,G,H$ all become forbidden. The set $\{ A,B,C \}$ remains.

\textit{Step 2 ---} The initial forbiddance will induce the condensation of $B$, which is accompanied by the partial forbiddance of $C$ (this is the same process as detailed in Sec.~\ref{subsubsec:DZ3}).

\textit{Step 3 ---} Then, anyon $C$ will split up into two $1$-dimensional anyons, $C_a$ and $C_b$. With the label-correspondence
\bq
A &=& 1 \\
C_a &=& e_1 \\
C_b &=& e_2
\eq
this forms the charge sector of $\mathcal{D}(\mathbb{Z}_{3})$: $\mathbb{Z}_{3}$. For an overview of this process, see Fig.~\ref{fig:Cx_Cy_limit}.

\begin{figure*}
\begin{centering}
\includegraphics[width=\textwidth]{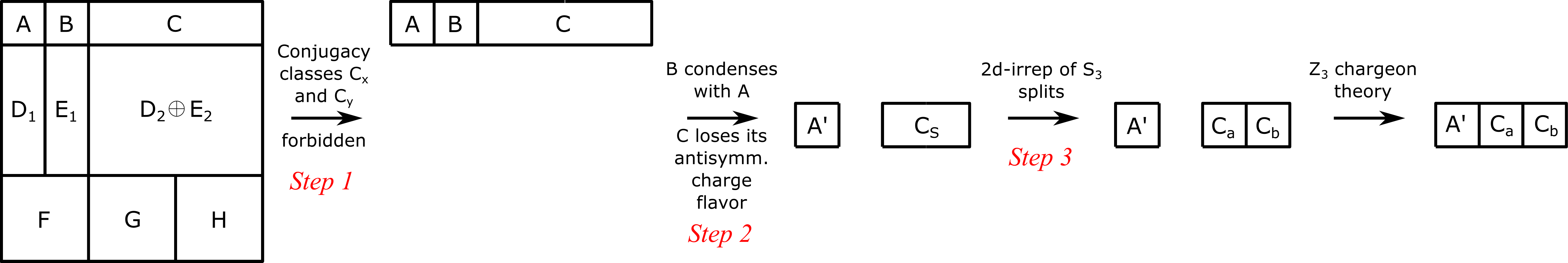}
\caption{Forbidding conjugacy classes $C_x$ and $C_y$ transforms $\mathcal{D}(S_3)$ to $\mathbb{Z}_2$. The subscript $S$ denotes the symmetric subspace of anyon $C$, and subscripts $a$ and $b$ denote the final, $1$-dimensional components. $A'$ is the vacuum after the condensation of $B$. Notice, that if we additionally forbid the irrep $\Gamma_{-1}$ of $S_3$, only \textit{Step 2} changes: instead of the induced condensation of anyon $B$, it is immediately forbidden along with the antisymmetric flavor of $C$. Afterwards, the two processes continue the same way.}
\label{fig:Cx_Cy_limit}
\end{centering}
\end{figure*}

\subsubsection*{Proof}

This proof follows the general ideas of the proof in Sec.~\ref{subsubsec:DZ3}; for more details on certain points please refer to that section.

After the initial forbiddances, the set $\{ A,B,C \}$ remains.

\textit{Step 1 ---} We truncate the fusion rules of $\mathcal{D}(S_3)$ to this set, these are shown in Table~\ref{tab:fusion_ABC}.
\begin{table}
	\centering
		\begin{tabular}{ |c||c|c|c| } 
			\hline
			 & A & B & C \\
			\hline
			\hline
			A & A & B & C \\ 
			\hline
			B & B & A & C \\ 
			\hline
			C & C & C & A $\oplus$ B $\oplus$ C \\
			\hline
		\end{tabular}
		\caption{Fusion rules in the chargeon sector of $\mathcal{D}(S_3)$.}
		\label{tab:fusion_ABC}
\end{table}

Then, using the new fusion rules to construct the $S$-matrix, we get:
\be \label{ABF_Smatrix}
S= \frac{1}{\sqrt{12}}\left[\begin{array}{ccc}
1 & 1 & 2 \\
1 & 1 & 2 \\
2 & 2 & -2
\end{array} \right] .
\ee

Now, we can state the following two lemmas:

\begin{lem}[Condensation (\textit{Step 2})] \label{lem:condensation_Z3}
Anyon $B$ condenses to the vacuum, and the following $S$-matrix describes the emerging theory:
\be \label{eq:Z3_matrix}
\frac{1}{\sqrt{3}}\left[\begin{array}{cc}
1 & \sqrt{2} \\
\sqrt{2} & -1 
\end{array} \right] .
\ee
\end{lem}

\begin{proof}
The proof follows that of Lemma~\ref{lem:condensation_DZ3}.
\end{proof}

\begin{lem}[Splitting of anyons (\textit{Step 3})] \label{lem:splitting_Z3}
Merging anyons of $\mathbb{Z}_3$ (chargeons $\{ e_1,e_2 \}$) yields a block-diagonal $S$-matrix, the physical block of which is \eqref{eq:Z3_matrix}.
\end{lem}

\begin{proof}
The $S$-matrix of $\mathbb{Z}_3$ is \cite{Bonderson07}:
\be
S= \frac{1}{\sqrt{3}}\left[\begin{array}{ccc}
1 & 1 & 1 \\
1 & \omega & \omega^2 \\
1 & \omega^2 & \omega
\end{array} \right] .
\ee
Using this $S$-matrix, we follow the proof steps of Lemma~\ref{lem:splitting_DZ3}.
\end{proof}

Combining Lemmas~\ref{lem:condensation_Z3} and \ref{lem:splitting_Z3} we conclude that $\mathbb{Z}_3$ indeed emerges, after the condensation of $B$ and split-up of anyon $C$. \qed
\vspace{0.5cm}

\subsection{Forbidding $C_x$ and $\Gamma_2$ jointly leads to $\mathbb{Z}_3$}

Forbidding $C_x$ and $\Gamma_2$ results in:

\textit{Step 1 ---} Anyons $C,D,E,G,H$ become forbidden, and the set $\{ A,B,F \}$ remains.

The original $\mathcal{D}(S_3)$ had a $C \leftrightarrow F$ exchange symmetry, thus it is easy to see that the physical processes this transition will follow will be identical to the ones detailed in Sec.~\ref{subsec:Z3_C}.

\textit{Step 2 ---} The condensation of $B$ will follow, $F$ becomes partially forbidden.

\textit{Step 3 ---} Then, $F$ splits up into two $1$-dimensional anyons, which together with the vacuum will form $\mathbb{Z}_{3}$. The resulting $\mathbb{Z}_{3}$ now is the flux sector of $\mathcal{D}(\mathbb{Z}_{3})$, as $F$ is a fluxon.

This process is shown in Fig.~\ref{fig:Cx_Gamma2_limit}.

\begin{figure*}
\begin{centering}
\includegraphics[width=\textwidth]{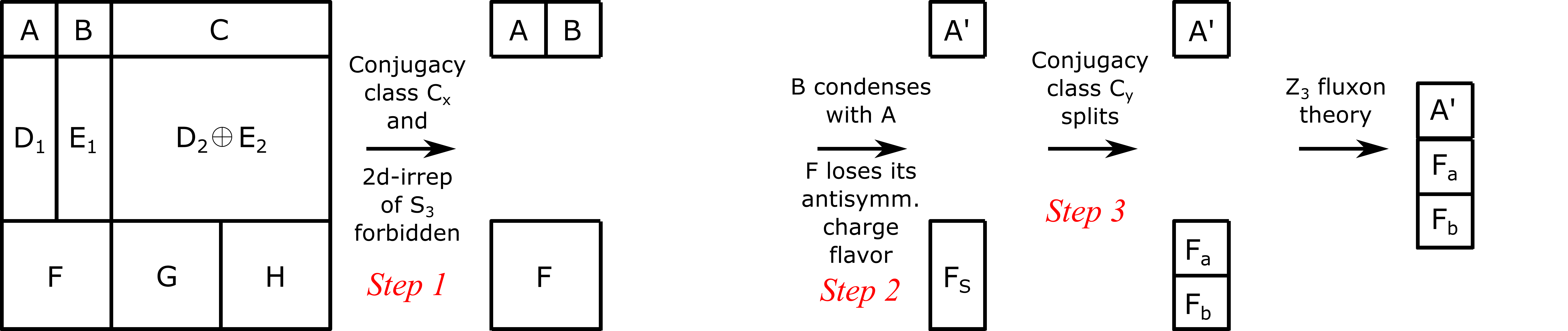}
\caption{Forbidding conjugacy class $C_x$ and irrep $\Gamma_2$ of $S_3$ transforms $\mathcal{D}(S_3)$ to $\mathbb{Z}_2$. The subspace $F_S \equiv F_1$ of Fig.~\ref{fig:anyon-splitting-2}, and subscripts $a$ and $b$ denote the final, $1$-dimensional components. $A'$ is the vacuum after the condensation of $B$. If we additionally forbid the irrep $\Gamma_{-1}$ of $S_3$, only \textit{Step 2} changes: instead of the induced condensation of anyon $B$, it is immediately forbidden along with the nontrivial charge flavor of $F$. After this step, the processes are the same.}
\label{fig:Cx_Gamma2_limit}
\end{centering}
\end{figure*}

\subsubsection*{Proof}

The proof is identical to the one presented in Sec.~\ref{subsec:Z3_C}. \qed

\subsection{Forbidding $\Gamma_{-1}$, $C_x$, $C_y$ jointly leads to $\mathbb{Z}_3$} \label{subsec:Z3_2_C}

Forbidding $\Gamma_{-1}$, $C_x$ and $C_y$ will lead to:

\textit{Step 1 ---} The forbiddance of $B,D,E,F,G,H$, leaving only anyons $\{ A,C \}$.

\textit{Step 2 ---} The forbiddance of $B$ induces the partial forbiddance of $C$ (similar to how condensation of $B$ led to the forbiddance of all anyons' antisymmetric charge flavors, see above, or Sec.~\ref{subsubsec:DZ3}).

\textit{Step 3 ---} Then, $C$ will split up and form the charge sector of $\mathcal{D}(\mathbb{Z}_{3})$.

\subsubsection*{Proof}

The proof of this is straightforward, one only has to dutifully follow the protocol detailed in Sec.~\ref{subsec:protocol}. The $S$-matrix the protocol yields is identical to Eq.~\eqref{eq:Z3_matrix}. For how to model the split-up of anyon $C$, simply consult the proof of Lemma~\ref{lem:splitting_Z3}. \qed

\subsection{Forbidding $\Gamma_{-1}$, $\Gamma_2$, $C_x$ jointly leads to $\mathbb{Z}_3$}

Processes similar to those of Sec.~\ref{subsec:Z3_2_C} happen when we forbid the combination $\Gamma_{-1}$, $\Gamma_2$ and $C_x$:

\textit{Step 1 ---} Anyons $B,C,D,E,G,H$ become forbidden; as well as 

\textit{Step 2 ---} $F$ is partially forbidden (only the charge flavor $\Gamma_1$ remains). This leaves the set $\{ A,F \}$. 

\textit{Step 3 ---} Then, anyon $F$ splits into two particles, forming the flux sector of $\mathcal{D}(\mathbb{Z}_{3})$.

\subsubsection*{Proof}

The proof is identical to the one presented in Sec.~\ref{subsec:Z3_2_C}. \qed

\subsection{Forbidding $C_x$ and $\Gamma_{-1}$ jointly leads to $\cD(\mathbb{Z}_3)$}  \label{subsec:DZ3_2}

By forbidding the conjugacy class $C_x$ and the $\Gamma_{-1}$ irrep at the same time:

\textit{Step 1 ---} Anyons $B,D,E$ are forbidden, the remaining set of anyons is $\{A, C, F, G, H\}$.

\textit{Step 2 ---} The forbiddance of anyon $B$ induces $C,F,G,H$ losing their antisymmetric charge flavors.

\textit{Step 3 ---} Then $C,F,G,H$ will split up into two each, similar to the process detailed in Sec.~\ref{subsubsec:DZ3}.

This process is summarized in Fig.~\ref{fig:Cx_limit} of the main text.

\subsubsection*{Proof}

We need to truncate the fusion rules to the remaining set of anyons $\{ A,C,F,G,H \}$, then construct the $S$-matrix. The matrix this protocol yields is identical to Eq.~\eqref{DZ3_correspondence_mtx}. From there, we can follow the steps detailed in the proof of Sec.~\ref{subsubsec:DZ3}, to prove that all excitations will split up, to form the excitations of $\mathcal{D}(\mathbb{Z}_3)$. \qed

\end{document}